\documentclass{article}

\PassOptionsToPackage{numbers}{natbib}
\usepackage[main, final]{neurips_2026}

\usepackage[utf8]{inputenc} 
\usepackage[T1]{fontenc}    
\usepackage{hyperref}       
\usepackage{url}            
\usepackage{booktabs}       
\usepackage{amsfonts}       
\usepackage{nicefrac}       
\usepackage{microtype}      
\usepackage{xcolor}         

\usepackage{algorithm}
\usepackage{algorithmic}
\usepackage{amsmath}
\usepackage{amsthm}
\usepackage{amssymb}
\usepackage{booktabs}
\usepackage{enumerate}
\usepackage{xcolor}
\usepackage{tikz}

\usepackage{tikz} 
\usetikzlibrary{arrows.meta}
\usetikzlibrary{positioning, arrows.meta, calc}

\newtheorem{definition}{Definition}
\newtheorem{lemma}{Lemma}
\newtheorem*{lemma*}{Lemma}
\newtheorem{theorem}{Theorem}
\newtheorem{proposition}{Proposition}
\newtheorem{corollary}{Corollary}

\usepackage{subcaption}

\title{When Everyone Can Submit: \\ Designing Contests with Transparent Pre-Selection}

\author{%
  Hanbing Liu \\
  Gaoling School of Artificial Intelligence\\
  Renmin University of China\\
  \texttt{liuhanbing@ruc.edu.cn} \\
  \And
  Ningyuan Li \\
  CFCS, School of Computer Science\\
  Peking University\\
  \texttt{liningyuan@pku.edu.cn} \\
  \AND
  Weian Li \\
  School of Software\\
  Shandong University\\
  \texttt{weian.li@sdu.edu.cn} \\
  \And
  Qi Qi\thanks{Corresponding author.} \thanks{Also with, Beijing Key Laboratory of Research on Large Models and Intelligent Governance. Also with, Engineering Research Center of Next-Generation Intelligent Search and Recommendation, MOE.} \\
  Gaoling School of Artificial Intelligence\\
  Renmin University of China\\
  \texttt{qi.qi@ruc.edu.cn} \\
  \AND
  Changyuan Yu \\
  Baidu Inc.\\
  \texttt{yuchangyuan@baidu.com} \\
}

\ifdefined\XeTeXversion
  \catcode`’=\active \def’{'}
  \catcode`—=\active \def—{---}
\fi

\begin{document}

\maketitle

\begin{abstract}
Shortlisting is a common and effective method for pre-selecting participants in competitive settings. To ensure fairness, a cut-off score is typically announced, allowing only contestants who exceed it to enter the contest, while others are eliminated. In this paper, we study rank-order contests with shortlisting and cut-off score disclosure. We fully characterize the equilibrium behavior of shortlisted contestants for any given prize structure and shortlist size. We examine two objective functions: the highest individual performance and total performance. Under linear performance costs, for both objectives, the optimal contest is in a winner-take-all format. For the highest individual performance, the optimal shortlist size is exactly two contestants, but, in contrast, for total performance, the shortlist size does not affect the outcome, i.e., any size yields the same total performance. Furthermore, we compare the highest individual performance achieved with and without shortlisting, and show that the former is \(4/3\) of the latter.
    
\end{abstract}

\section{Introduction}
\label{sec:intro}

Recent advances in generative AI have substantially lowered the barrier to participation in many online production markets. Individuals can now rely on highly capable large language models based coding agents to rapidly produce submissions for coding, design, writing, and data-science tasks. As a result, contest-based crowdsourcing platforms such as Kaggle and Freelancer have experienced a sharp increase in participation intensity (See Figure~\ref{fig:banner}). 

\begin{figure}[h]
\centering
\begin{subfigure}[h]{0.3\linewidth}
    \centering
    \includegraphics[width=\linewidth]{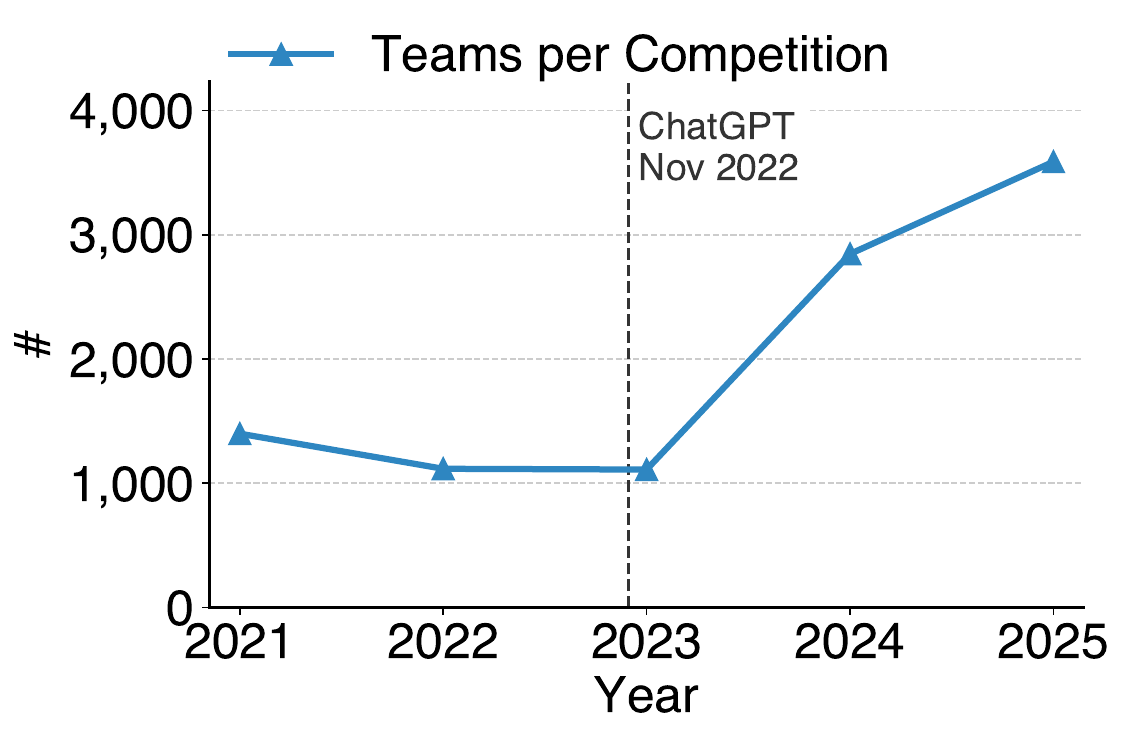}
    \caption{Kaggle}
\end{subfigure}
\begin{subfigure}[h]{0.3\linewidth}
    \centering
    \includegraphics[width=\linewidth]{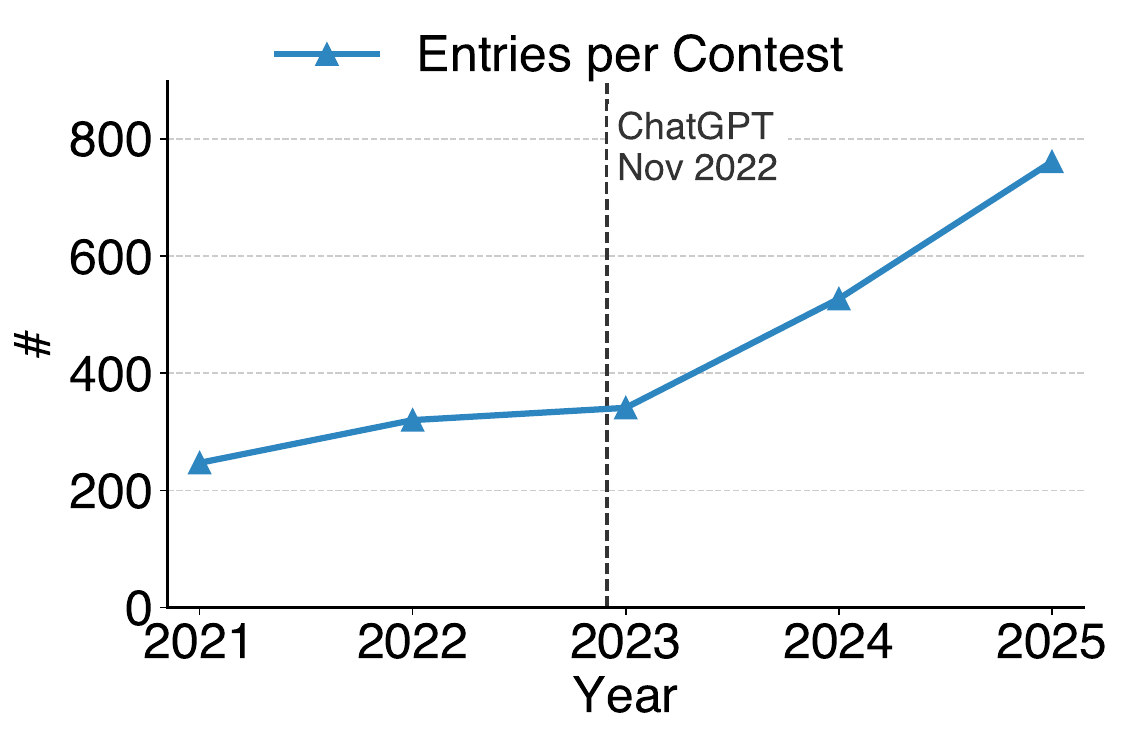}
    \caption{Freelancer}
\end{subfigure}
\caption{\textbf{After the emergence of ChatGPT, the average level of competition on online contest and crowdsourcing platforms increased sharply}, as evidenced by the Kaggle Playground Series and contests on Freelancer (data sources: Kaggle API and Freelancer annual reports).}
\label{fig:banner}
\end{figure}

While broader participation improves accessibility, it also creates severe congestion for platform operators. 
With each task attracts a massive number of submissions, evaluating all entries becomes costly and slow, increasing burdens on review pipelines and computational resources.
Therefore, in practice, many platforms already rely 
on qualification rounds or historical ratings to shortlist participants to reduce overhead, 
and a cut-off score is disclosed to promote fairness and transparency.
This practice motivates several interesting theoretical questions. Classic contest models assume full participation, while allowing shortlisting might lead to different game structures. 
How do the pre-selection process affect contestants' performance incentives? Moreover, how should the designer optimally jointly choose the prize structure and shortlist size to best achieve the intended objective?

There are no obvious answers. On the one hand, shortlisting and cut-off disclosure may intensify competition among shortlisted contestants, thereby increasing effort. On the other hand, excluding too many participants may weaken incentives, as a smaller pool of competitors makes prizes easier to obtain. Understanding these opposing forces requires a precise characterization of equilibrium performance given the shortlist size, prize structure, and information disclosure policy.

From the designer’s perspective, shortlisting substantially complicates the optimization problem relative to standard settings in which all contestants compete. Consider, for example, a designer who seeks to maximize total performance (e.g., \cite{MS01}). While reducing the number of finalists may raise individual effort through intensified competition, total performance need not increase monotonically, since fewer contestants contribute effort overall. Consequently, existing analytical methods do not directly extend to environments with shortlists and cut-off disclosure. New theoretical tools are therefore required to study the optimal contest design in this more complex setting.

\subsection{Main Contributions}

\textbf{Model.} 
To address these challenges, we analyze a rank-order contest incorporating a shortlist. Suppose there are $n \geq 2$ initial registrants, each with a private ability drawn from a public distribution.
The contest executes in two stages. In the first stage, the designer selects the number of shortlisted contestants $m$ and determines a prize structure, subject to a total budget. All $n$ contestants are then ranked by ability in descending order, and the top $m$ are shortlisted. The cut-off score, i.e., the ($m+1$)-th ability score, is disclosed to the admitted contestants.
In the second stage, shortlisted contestants, now aware of the prize structure and the borderline ability, choose how much effort to deliver. Their chosen performance levels, which incur costs, depend on their private abilities. Each contestant aims to maximize the difference between the prize received and the cost of performance.

\textbf{Contest Design.} Based on this contest-theoretic model, we first show that, for any given shortlist size and prize structure, contestants’ equilibrium behavior can be uniquely and tractably characterized.
We then study two common design objectives and show that, in both cases, the optimal prize structure is winner-take-all. However, the role of shortlisting differs sharply across objectives: when maximizing the highest individual performance, the optimal shortlist consists of exactly two contestants, whereas for total performance, the outcome is invariant to the shortlist size. Moreover, the highest individual performance achieved under shortlisting is \(4/3\) times that attainable without shortlisting. Finally, we complement our theoretical findings with numerical experiments.


To illustrate the non-trivial technical tools underlying our results, we provide an overview of our approach to derive the highest individual performance. The procedures are shown in Figure~\ref{fig:Overview}.

\begin{figure*}[t!]
    \centering








\resizebox{\linewidth}{!}{%
\begin{tikzpicture}[
  font=\small,
  boxw/.store in=\boxw, boxw=1.5cm,       
  boxh/.store in=\boxh, boxh=1.0cm,       
  gap/.store in=\gap,   gap=1.6cm,        
  theorembox/.style={
    draw, rounded corners=2pt, very thick,
    minimum width=\boxw, minimum height=\boxh,
    align=center, inner sep=6pt
  },
  thmlbl/.style={
    font=\footnotesize, inner sep=2pt
  },
  lemmalbl/.style={
    font=\footnotesize, inner sep=1pt
  },
  arr/.style={
    very thick, -{Latex[length=2.6mm, width=1.8mm]}
  }
]

\node[theorembox] (b1) {
  \((m,\vec{V}) \in \mathbb{N} \times \mathbb{R}_{+}^n\)
};
\node[theorembox, right=\gap of b1] (b2) {
  $(m,l) \in \mathbb{N} \times \mathbb{N}$
};
\node[theorembox, right=\gap of b2] (b3) {
  $m \in \mathbb{N}$
};
\node[theorembox, right=\gap of b3] (b4) {
  $k \in [0,1)$
};
\node[theorembox, right=\gap of b4] (b5) {
  $m^* = 2, V^*_1=B$ 
};

\node[thmlbl, above=1mm of b1] {Thm.~\ref{thm:sbne}};
\node[thmlbl, above=0.6mm of b2] {Prop.~\ref{prop:OptSimple}};
\node[thmlbl, above=0.6mm of b3] {Prop.~\ref{prop:HighestWTA}};
\node[thmlbl, above=1mm of b4] {};
\node[thmlbl, above=1mm of b5] {Thm.~\ref{thm:AsmHP}};

\draw[arr] (b1.east) -- (b2.west);
\node[lemmalbl] at ($(b1.east)!0.5!(b2.west) + (0,3mm)$) {Coro.~\ref{coro:ZeroPrize}};

\draw[arr] (b2.east) -- (b3.west);
\node[lemmalbl, align=center] at ($(b2.east)!0.5!(b3.west) + (0,5mm)$) {Lem.~\ref{lem:Truncate}, \ref{lem:iso1}\\Lem.~\ref{lem:IntDerMon}, \ref{lem:betaInt}};

\draw[arr] (b3.east) -- (b4.west);
\node[lemmalbl] at ($(b3.east)!0.5!(b4.west) + (0,3mm)$) {Lem.~\ref{lem:iso2},~\ref{lem:betarep}};

\draw[arr] (b4.east) -- (b5.west);
\node[lemmalbl] at ($(b4.east)!0.5!(b5.west) + (0,3mm)$) {Lem.~\ref{lem:AsmStep}};

\end{tikzpicture}
}
    \caption{\textbf{Technical Overview} for solving the highest individual performance contest design problem. Each box denotes the decision variables of the current optimization problem. Starting from a discrete high-dimensional continuous program over shortlist size and prize vectors, we progressively simplify the problem through a sequence of technical lemmas into an asymptotic single-variable continuous optimization. This reduction enables a distribution-independent characterization of the optimal contest, yielding a 2-contestant winner-take-all design.}
    \label{fig:Overview}
    \vspace{-1em}
\end{figure*}
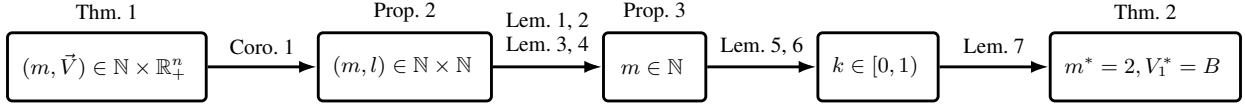



\noindent\textbf{1. Prize Structure Guideline.}
A key technical insight is that optimal prize structures can be significantly simplified under linear cost. For any objective that is linear in contestants’ performances, the optimal contest must be a \emph{simple contest}, in which all non-zero prizes are equal (Proposition~\ref{prop:OptSimple}). This reduces the designer’s problem to choosing only the shortlist size and the number of prizes. 

\noindent\textbf{2. Winner-take-all Format.}
Fix a shortlist size. The designer’s objective equals the expected equilibrium performance of the strongest admitted contestant, which initially depends on the full prize vector and conditional order statistics. Using a probabilistic identity, the conditional structure is replaced by order statistics drawn from an i.i.d. truncated distribution, eliminating dependence on the original high-dimensional conditioning event (Lemma~\ref{lem:Truncate}).
Equilibrium performance then becomes additively separable across prize differences, reducing the continuous prize vector to independent rank-wise marginal contributions. This collapses the designer’s optimization to selecting the single rank with the largest impact. After transforming into quantile space (Lemma~\ref{lem:iso1}), these contributions become distribution-free and strictly ordered across ranks (Lemma~\ref{lem:IntDerMon},~\ref{lem:betaInt}), showing that only the top prize matters. Hence the optimal structure degenerates to winner-take-all (Proposition~\ref{prop:HighestWTA}).

\noindent\textbf{3. Asymptotic Analysis via Beta Representation.} 
With the prize dimension eliminated, the remaining variable is the shortlist size. The expected highest performance still involves nested integrals over truncated order statistics. By rewriting the objective to separate distribution-dependent terms from a universal kernel, all dependence on the ability distribution is pushed into a single outer integral (Lemma~\ref{lem:iso2}).
The kernel is then represented using regularized beta functions that depend only on the shortlist size and the cut-off position (Lemma~\ref{lem:betarep}). As the number of contestants grows, the beta term concentrates at the boundary, causing the entire stochastic structure to collapse to a closed-form limit depending solely on the shortlist size (Lemma~\ref{lem:AsmStep}).
This reduces the original joint optimization to a one-dimensional comparison across shortlist ratio, whose maximizer is $m^*=2$ (Theorem~\ref{thm:AsmHP}).

Due to space limitation, 
all missing proofs
can be found in the appendix.

\section{Related Literature}
\textbf{Contest Design.} Our study falls within the domain of single 
rank-order contest (a.k.a., all-pay contests) design. The existing literature primarily examines contestants’ equilibrium behavior and the design of optimal prize structures for various designer objectives. Foundational research has analyzed equilibria in all-pay auctions \cite{BKV96, BK98, BKV12}.
Regarding contest design, studies have explored different objective functions, including total performance \citep{GH88, MS01, KG03}, performance of the top-\(k\) contestants \citep{AS09}, highest individual performance \citep{CHS19}, and threshold-based objectives \citep{EGG21}. Negative prizes are examined by \citep{LLWZ18} and \citep{LL23}. 
\citep{G23} analyzes how prize structures and ability distributions shape equilibrium efforts.
Closer to our setting, \citep{SSYJ24} and \citep{LLLQY25} investigate rank-order contests with shortlists, aiming to optimize total or highest individual effort. However, their models often lack guarantees of equilibrium existence.
In contrast, we propose a rank-order contest that incorporates both shortlisting and cut-off score disclosure, and ensure the existence of equilibrium in all instances. 

\textbf{Elimination and Signaling in Contests.} In addition, the disclosure of borderline ability links our work to the literature on elimination and signaling in contests. The seminal paper \cite{MS06} analyzes two-stage contests where winners of preliminary sub-contests advance to a final round. \citep{FL12} study optimal multi-stage Tullock contest designs \cite{T08}, while \cite{LMZ18} and \cite{LSA19} examine how information disclosure influences contestant behavior. \cite{MPS21} investigate sequential all-pay elimination contests under complete information, and \cite{FW22} incorporate bias and information disclosure into two-stage settings. \cite{R24} explores a multi-stage format in which the lowest-effort contestant is eliminated at each step.
In parallel, research on Bayesian persuasion in contest design has emerged \cite{CKZ17, C24, KZZ24}, providing insights into optimal information disclosure strategies.
Distinct from these lines of work, our model introduces a single elimination stage through shortlisting, followed by a simple signal indicating whether a contestant qualifies and the corresponding borderline ability. 
\section{Model and Preliminaries}
\label{sec:pre}

\subsection{Contest with Transparent Pre-Selection}
In this section, we formally introduce the contest model with transparent shortlisting.
Consider a single-contest setting with \(n\) registered contestants. Each contestant $i \in [n]$ has a finite private ability level $x_i\in (0,\bar{x}]$, independently and identically drawn from a publicly known distribution $F(\cdot)$, which has a continuous and strictly positive density function $f(\cdot)>0$. For convenience, we also represent ability in terms of the quantile $q_i := 1-F(x_i)$ which follows a uniform distribution on $[0, 1]$, i.e., $q_i \sim U[0,1]$. The inverse function, $x(q_i):=F^{-1}(1-q_i)=x_i$ is strictly decreasing, implying that lower quantiles correspond to higher ability levels. We also assume $x'(\cdot)$ to be bounded.

As outlined in Section \ref{sec:intro}, we consider contests with a pre-selection stage. The contest designer first determines the shortlist capacity, i.e., the number of admitted contestants, denoted by \( m\), where $2\leq m \leq n-1$, and specifies a prize structure $\vec{V}=(V_1, V_2, \cdots, V_m)$ within a total budget $B$. The prizes satisfy $V_1 \geq V_2\geq  \cdots \geq V_m$ and $\sum_{i=1}^m V_i \leq B$. We denote the designer’s decision, i.e., the contest configuration,  by $\mathcal{C}=(m,\vec{V})$. 
Then, the designer evaluates the abilities of all registered contestants,\footnote{In practice, a contestant’s ability may be inferred from resumes, historical performance records, or other relevant indicators. We assume the designer can accurately assess these abilities prior to the contest. In Section~\ref{sec:exp}, we show that our results are robust to noisy assessment through numerical experiments.} selects the top $m$ contestants based on ability, and admits them to the shortlist, eliminating the rest. The designer also publicly discloses the borderline ability, which is the cut-off score,\footnote{This captures information asymmetry that prevails online competition platforms. For instance, on Kaggle, designers can evaluate contestants based on profiles and past performance. Rating statistics are public available, and during registration, organizers can view the ratings of all registrants, whereas individual contestants only observe the total number of registrants.} defined as the ability of the ($m+1$)-st ranked contestant, denoted by $x_{(m+1)}$.\footnote{To balance readability and formality, we use both $x_{(i)}$ and $x^{(n-i+1)}$ to denote the ability of the \(i\)-th ranked contestant. The subscript notation is more intuitive, while the superscript notation aligns with the formal definition of order statistics.} 
Importantly, the shortlist size \( m \) and prize structure \( \vec{V} \) are chosen before the designer observes the actual abilities. These decisions depend only on the total number of registered contestants $n$ and the distribution $F$.  

Next, each admitted contestant $i$, i.e., one whose ability $x_i$ ranks among the top $m$, is notified and then strategically chooses a performance level $e_i$ to compete for prizes. The cost of performance is given by $g(e_i)/x_i$, where $g(\cdot): \mathbb{R}_{\geq 0} \rightarrow \mathbb{R}_{\geq 0}$ is a strictly increasing, continuous, and differentiable function with $g(0)=0$.\footnote{Linear costs are widely adopted in both theory and practice to model outputs in contests. For instance, in all-pay auctions, often viewed as a special case of contests, payments follow a linear cost structure. In political settings, cost may represent campaign investment. Additionally, prior studies such as \cite{MS01, AS09, DV09} have explored contests under linear cost assumptions.} Intuitively, for a given performance $e$, a contestant with higher ability incurs a lower cost than one with lower ability. After all performance levels are chosen, prizes are awarded according to a rank-order rule: the contestant executing the highest performance receives prize $V_1$, the second-highest performance earns $V_2$, and so on. 


Given our contest model, we now define each contestant’s utility. Eliminated contestants do not perform and receive no prize, so their utility is trivially zero.
Let $\mathbf{e}= (e_1, e_2,\cdots, e_m)$ denote the performance profile of all admitted contestants. For an admitted contestant $i$, utility is defined as the prize received minus the cost of performance:
\[
    u_i(\mathbf{e}) = V_{Rank(i,\mathbf{e})}-\frac{g(e_i)}{x_i},
\]
where $Rank(i,\mathbf{e})$ denotes contestant $i$'s rank in the performance profile $e$. Each contestant chooses their performance level to maximize this utility.

For each admitted contestant $i$, the performance $e_i$ depends solely on their ability $x_i$ and a strategy function $b_i(x): \mathbb{R}_{\geq 0} \rightarrow \mathbb{R}_{\geq 0}$,  such that $e_i=b_i(x_i)$. 
Given the incomplete information setting, we adopt the standard symmetric Bayesian Nash equilibrium (sBNE) as our solution concept. Intuitively, because admitted contestants are ex ante identical and face the same strategic environment, their strategic incentives are expected to align around a common equilibrium strategy function. Formally,
\begin{definition}\label{def:sBNE}
    A strategy function constitutes a symmetric Bayesian Nash equilibrium if and only if, for any admitted contestant, following this strategy maximizes the expected utility given that all other contestants follow the same strategy. Specifically, 
    for any admitted contestant $i$,
\[
    b^*(x_i) \in \arg\max_{e_i} \sum_{j=1}^m V_j \cdot \text{Pr}_{ij}(e_i) -\frac{g(e_i)}{x_i},
\]
    where $\text{Pr}_{ij}(e_i)$ denotes the probability that contestant $i$ believes that she ranks $j$-th in the contest by performing \( e_i \), assuming all other contestants follow the strategy $b^*(x)$. 
\end{definition}

Finally, the contest designer seeks to optimize an objective function by choosing the contest configuration $\mathcal{C}=(m, \vec{V})$, which specifies the shortlist size and the prize structure. In this work, we focus on two widely studied and practically useful objectives: the highest individual performance and the total performance. Specifically, given a contest configuration $\mathcal{C}$, the highest individual performance objective is defined as
\(
    \textstyle \text{HP}(\mathcal{C},F)= \mathbb{E}_{x_1, x_2,\cdots, x_n \sim F}\big[e_{(1)}\big],
\)
where $e_{(1)}$ denotes the highest performance among all contestants under the sBNE. The total performance objective is defined as
\(
    \textstyle \text{TP}(\mathcal{C},F)= \mathbb{E}_{x_1, x_2,\cdots, x_n \sim F}\big[\sum_{i=1}^n e_{i}\big],
\)
with that eliminated contestants execute no performance.


\subsection{Simple Contests}
We introduce a special class of contests that plays a central role in the optimal contest design.
\begin{definition}[Simple Contest]
    A contest with $n$ contestants and shortlist size $m$ is a simple contest if all non-zero prizes are equal. If all prizes are equal, i.e., either $V_i = 0$ for all $i$, or $V_1 = \ldots = V_m \neq 0$, the contest is a trivial simple contest.
\end{definition}
A simple contest with exactly one positive prize is referred to as a \emph{winner-take-all} contest.
\begin{definition}[Winner-Take-All Contest]
    A contest with $n$ contestants and shortlist size $m$ is called a winner-take-all contest if there is only one non-zero prize, i.e.,  $V_1 \neq 0$ and $V_i =  0$ for \(i\in\{2,\cdots,m\}\).
\end{definition}





\subsection{Contestant Equilibrium}

In this subsection, we characterize contestants’ equilibrium
behavior under a given contest configuration $\mathcal{C} = (m, \vec{V})$. Eliminated contestants do not participate and therefore have no performance. For admitted contestants, however, the announcement of the cut-off score invalidates prior beliefs about opponents’ abilities, requiring each contestant to update her posterior beliefs accordingly.

\begin{proposition}[Posterior Beliefs]
For any admitted contestant $i$, conditional on the highest ability among the eliminated contestants, denoted by $x^{(n-m)}$, the posterior beliefs about the abilities of the remaining contestants are mutually independent and do not depend on contestant $i$’s own ability. Specifically, for any other admitted contestant $j \in [m] \setminus \{i\}$, the posterior cumulative distribution function (CDF) and probability density function (PDF) of ability are given by the truncated distribution:
\(
    \textstyle P(x) = \frac{F(x) - F(x^{(n-m)})}{1 - F(x^{(n-m)})},\, p(x) = \frac{f(x)}{1 - F(x^{(n-m)})}.
\)
\end{proposition}

Given the posterior belief, the stage after pre-selection can be viewed as a contest with \(m\) contestants, each drawing ability from the posterior \( P(x) \). As a result, we establish the existence and uniqueness of the sBNE, given as follows.

\begin{theorem}[Unique sBNE]\label{thm:sbne}
    For any contest $\mathcal{C}=(m, \vec{V})$,  any ability distribution \( F \), and any strictly increasing cost function \( g \), there exists a unique symmetric Bayesian Nash equilibrium among the admitted contestants, given the cut-off score \( x^{(n-m)} \), which is characterized by the following:
\[
\begin{aligned}
    b(x;x^{(n-m)}) =  g^{-1} \Big ( \textstyle \int_{x^{(n-m)}}^x \sum_{l=1}^{m-1}(V_l-V_{l+1})\binom{m-1}{l-1}
    \cdot(m-l)(1-P(t))^{l-1}P(t)^{m-l-1}p(t)t\, dt \Big ).
\end{aligned}
\]
\end{theorem}

\section{Optimal Contest Design}
\label{sec:opt}
In contest design problems, deriving equilibrium is merely the first step. Different from traditional rank-order contests, with a pre-selection stage, the designer faces a much harder high-dimensional optimization problem: the joint choice of discrete shortlist size $m$ and continuous prize structure $\vec{V}$. This added decision variable increases both the dimensionality and complexity of the optimization problem. In this section, we first introduce a prize structure design guideline, then present the optimal contest design for maximizing the highest individual performance and total performance.
\subsection{Prize Structure Design Guideline}

To simplify the contest design process, 
we establish two key properties of equilibrium performance from the designer’s perspective, which help reduce the dimensionality of decision space. First, offering a consolation prize (i.e., $V_m\neq0$) weakens the overall incentive by reducing the prize gap.
\begin{corollary}[No Consolation Prize Should be Set]\label{coro:ZeroPrize}
    If a contest \( (m, \vec{V}) \) awards
    the last-place contestant, i.e., \( V_m > 0 \), then setting \( V_m = 0 \) increases the equilibrium performance of all contestants.
\end{corollary}

We now formally state the second property.

\begin{proposition}[The Optimal Contest is a Simple Contest] \label{prop:OptSimple}
    If the designer's objective is a linear combination of the  admitted contestants' performances, and the cost function is linear \( g(e) = ke \) for some \( k > 0 \), then the optimal contest that maximizes the designer’s objective must be a non-trivial simple contest, regardless of whether the objective is evaluated ex ante (in expectation), \( u_d = \mathbb{E}[ \vec{c} \cdot e(\vec{x})] \), or ex post (given a realized ability profile), \( u_d(\vec{x}) = \vec{c} \cdot e(\vec{x}) \).
\end{proposition}
The intuition behind Proposition \ref{prop:OptSimple} is that, under the stated conditions, the designer's objective becomes linear in the gaps between adjacent prizes. As a result, the optimal strategy is to allocate the budget to maximize the gap between the most influential adjacent prize levels. Given the rank-order constraint, the most efficient approach is to equalize all higher-tier prizes, yielding a simple contest.

Proposition \ref{prop:OptSimple} significantly simplifies the contest design problem: the designer only needs to determine two variables: (1) the shortlist size and (2) the number of prize recipients, after which the total prize budget is evenly divided among winners at each selected rank.
In the remainder of this section, we analyze the optimal contest design for two commonly used linear objectives under a linear cost function: the highest individual performance and total performance.
For simplicity, we assume unit cost \( g(e) = e \) (i.e., \( k = 1 \)) and normalize the budget to \( B = 1 \). Since the designer’s utility is linear in both parameters, the results generalize to arbitrary \( k \) and \( B \) through a simple rescaling by \( k^{-1}B \).   


\subsection{Highest Individual Performance}
Next, we focus on maximizing the highest individual performance. By the monotonicity of equilibrium performance, this is achieved by the contestant with the highest ability. Fixing the shortlist size, we first determine the optimal prize structure and show that it is also winner-take-all.
\begin{proposition}[Optimal Prize Structure for the Highest Individual Performance] \label{prop:HighestWTA}
    Given the shortlist size \( m \), the optimal prize structure for maximizing the highest individual performance is winner-take-all, i.e., \( V_1 = 1 \) and \( V_l=0\) for all \( l>1\).
\end{proposition}

\begin{proof}[Proof Sketch]
    Given the number of admitted contestants \( m \) and the cut-off score \( y \), the highest performance can be expressed as \(\mathbb{E}_{x \sim X^{(n)} | X^{(n-m)}=y}[b(x;y)]\). The distribution conditional on order statistics values is hard to tackle, thus we first establish a probability identity to simplify that:

    \begin{lemma}[Truncated Distribution]\label{lem:Truncate}
Let \(Y_1, Y_2, \ldots, Y_m\) be i.i.d. random variables distributed according to \(F\) truncated on \((a, \infty)\), with density function 
\(f_Y(y) = \frac{f(y)}{1 - F(a)}, \quad y > a.\)
Set \(r = n-m\). When \(X^{(r)} = a\), the conditional distribution \(\left(X^{(n-m+1)}, X^{(n-m+2)}, \ldots, X^{(n)}\right) \mid X^{(n-m)} = a\) is identical to the distribution of \(\left(Y^{(1)}, Y^{(2)}, \ldots, Y^{(m)}\right)\). 
\end{lemma}

     By Lemma~\ref{lem:Truncate}, we can instead take the expectation over the order statistic \( P^{(m)} \) of the truncated distribution \( P(x) \), i.e., 
$\mathbb{E}_{x \sim P_y^{(m)}}[b(x;y)]$ expands to
\[
(V_l-V_{l+1})\int_0^\infty \left ( \int_0^x \sum_{l=1}^{m-1}  \textstyle (m-l)\binom{m-1}{l-1}
(1-P(t))^{l-1}P(t)^{m-l-1}p(t)t\, dt \right ) \,dP^{(m)}(x).
\]

Let \(Z_l = l(V_l - V_{l+1})\), after switching the integration order, we obtain an expression in terms of \( Z_l \):
\[
    \mathbb{E}_{x \sim P_y^{(m)}}[b(x;y)]= \sum_{l=1}^{m-1}Z_l \int_0^\infty\binom{m-1}{l}  [1-P(t)^{m}]
    (1-P(t))^{l-1}\cdot P(t)^{m-l-1}p(t)t \, dt.
\]
By Corollary~\ref{coro:ZeroPrize}, the optimality requires \(V_m = 0\), hence \(\sum_{l=1}^{m-1} Z_l = \sum_{l=1}^{m}V_l = 1\). Thus, the optimization problem reduces to determine the \(Z_{l}\) with the maximum coefficient. Therefore, we transform the expression into quantile space \( q = 1-P(t) \) to isolate the distribution-dependent part.
\begin{lemma}[Distribution Isolation]\label{lem:iso1}
The coefficient of \(Z_l\) is
\[
\textstyle \int_0^1 |x'(q)| \left( \int_0^q \binom{m-1}{l} [1-(1-t)^m]t^{l-1}(1-t)^{m-l-1} \,dt \right) \,dq,
\]
where \( x(q) \) represents the ability value of the contestant at quantile \( q \in [0,1] \), and the distribution-independent term is denoted as \( G_{l}(q) \) \(=\int_0^q g_l(t)\, dt \).
\end{lemma}



Since \( |x'(q)| \) is an arbitrary positive real-valued function, to prove that winner-take-all is the optimal prize structure for any distribution, it suffices to show that \( G_1(q) \geq G_l(q) \) for all \( q \in [0,1] \) and \( 1<l\leq m-1 \).
We show this statement holds by proving \( G_l(q) \geq G_{l+1}(q) \) through two steps.

\underline{Step 1:} We first show that \( G_{l}(q)/G_{l+1}(q) \) is decreasing in \( q \), which is from the following lemma:

\begin{lemma}\label{lem:IntDerMon}
If \( f,g > 0 \), \( (\frac{f(x)}{g(x)})' < 0 \) then \( ( \frac{\int_0^x f(t) dt}{ \int_0^x g(t) dt} )'<0 \).
\end{lemma}

Since \( g_l(t)/g_{l+1}(t) = \frac{l+1}{m-l-1} \cdot \frac{1-t}{t} \) is decreasing in \( t \), \( G_{l}(q)/G_{l+1}(q) \) is also decreasing.

\underline{Step 2:} Next, it remains to show the boundary condition \( G_{l}(1)/G_{l+1}(1) \geq 1 \), by the following tool:
\begin{lemma}\label{lem:betaInt}
For any \( x > 0 \), \( a,b \in \mathbb{N} \), it holds that:
\(
\int_0^{x}t^{a-1}(x-t)^{b-1}\, dt = \frac{(a-1)!(b-1)!}{(a+b-1)!}x^{a+b-1}.
\)
\end{lemma}
We can transform the integral into a fractional form:
\(
G_l(1) = \frac{1}{l}\left( 1 - \frac{(m-1)!(2m-l-1)!}{(m-l-1)!(2m-1)!} \right).
\)
By induction, we can further prove \( G_l(1) - G_{l+1}(1) \geq 0 \).
In conclusion, we have shown that for all \( q \in [0,1] \), \( G_l(q) \geq G_{l+1}(q) \) holds. Therefore, for any distribution \( |x'(q)| > 0 \) and \( l\neq 1 \), we have:
\(
\int_0^1 |x'(q)| G_1(q) \, dq \geq \int_0^1 |x'(q)| G_l(q) \, dq,
\)
coefficient is maximized at $l=1$.
This proves that the winner-take-all contest is always optimal.
\end{proof}

Knowing that the optimal contest takes form of winner-take-all prize structure, we next examine the optimal shortlist size. 
We show that to maximize the highest individual performance, the shortlist size is two in the asymptotic setting.
\begin{theorem}[Asymptotic Optimal Contest for Highest Individual Performance]\label{thm:AsmHP}
For any ability distribution \( F \), as the number of contestants grows large (i.e., \( n \rightarrow \infty \)), the optimal contest
that maximizes the highest individual performance is a two-contestant winner-take-all contest (\( m = 2 \) and \( V_1 = 1 \)) and the value of objective converges to \( \frac{2}{3} \bar{x} \), where \( \bar{x} \) is the supremum of the support of the ability distribution.
\end{theorem}

\begin{proof}[Proof Sketch]
    When the shortlist size is \(m\), we denote the expected highest performance under a winner-take-all prize structure by \(S^{(n)}(m,n,1)\), which is,
\(
\begin{aligned}
    \mathbb{E}_{y \sim X^{(n-m)}} \big [ \mathbb{E}_{x^{(n)} \sim X^{(n)} \mid y} [ b(x^{(n)} ; y) ] \big ].
\end{aligned}
\)
Upon expansion, this yields a highly complex expression that is intractable in its raw form:
\[
\begin{aligned}
     \int_0^\infty \left ( \int_y^\infty x (m-1) [1-P(x)^m] P(x)^{m-2} p(x)  \, dx \right ) 
    \textstyle \binom{n-1}{m} F(y)^{n-m-1}(1-F(y))^m f(y) \, dy.
\end{aligned}
\]

Next, we reduce the original expression to a tractable form by special functions with favorable asymptotic properties, through a series of technical lemmas. Leveraging this form, we can ultimately determine the optimal shortlist size in the asymptotic regime. The detailed procedure is as follows:

\underline{Step 1:} 
To derive general conclusions,
we first separate distribution-independent components.
\begin{lemma}[Distribution Isolation]\label{lem:iso2}
    \(S^{(n)}(m,n,1)\) is equal to,
    \begin{align*}
    \int_0^1 x'(w) \left ( \int_w^1 \int_0^v n (m-1) \binom{n-1}{m} \left[1-\left(\frac{v-u}{1-u}\right)^m\right](v-u)^{m-2} u^{n-m-1} (1-u) \, du  dv \right ) \, dw,
\end{align*}
where \( x'(w) \) represents the incremental increase in ability value per unit quantile change. The distribution-independent  component is denoted as \( H_m(w) \).
\end{lemma}

\underline{Step 2:} To facilitate further analysis, we simplify \( H_m(w) \) through variable substitutions.
\begin{lemma}[Beta Function Representation]\label{lem:betarep} The distribution-independent part \( H_m(w) \) can be simplified using the regularized incomplete beta function \( I_x(a, b) = \frac{\int_0^x t^{a-1}(1-t)^{b-1}\,dt}{\int_0^1 t^{a-1}(1-t)^{b-1} \, dt} \), into:
    \[
    \frac{m}{2m-1} - (m-1)\int_0^w I_{\frac{w-z}{1-z}}(n-m, m+1)[1-z^m]z^{m-2} \, dz.
    \]
\end{lemma}

\underline{Step 3:} We show that the regularized incomplete beta function exhibits favorable properties as \( n \to \infty \), allowing it to be upper bounded by a simple step function.
\begin{lemma}[Upper Bounded by Step Function]\label{lem:AsmStep}
    For the regularized incomplete beta function \( I_x(n-m,m+1) \) and any \( \epsilon > 0 \), \( \delta > 0 \), there exists an \( N \in \mathbb{N} \) such that when \( n > N \), for arbitrary \( m \geq 2 \), the following holds: \( I_x(n-m,m+1) \leq \epsilon \) for \( x \in [0, 1-\frac{m}{n-1}-\delta) \) and \( I_x(n-m,m+1) \leq 1 \) for \( x \in [1-\frac{m}{n-1}-\delta, 1] \).
\end{lemma}
Using this property, we can conduct asymptotic analysis on \( S^{(n)}(m,n,1) = \int_0^1 x'(w) H_m(w) \, dw \) using the reformulated \( H_m(w) \) based on \( I_x\). Ultimately, for any given distribution \( F \), as \( n \to \infty \), and any shortlist size scheme (e.g., fixed proportion advancement or fixed number advancement) the highest performance uniformly converges to:
\(
\frac{m}{2m-1} \bar{x}.
\)
Therefore, we conclude that the maximum is attained when \( m = 2 \), yielding a highest individual performance of \( \frac{2}{3} \bar{x} \).
\end{proof}

Also, previous work has shown that, without pre-selection, \( n \)-contestant winner-take-all contest is optimal
\citep{AS09}. Thus, we can make the following comparison, showing the power of pre-selection.

\begin{theorem}[Competitive Ratio of Pre-selection]
\label{thm:ratio-hp}
For any ability distribution \( F \), the highest performance achieved by the optimal contest with transparent pre-selection is asymptotically \( \frac{4}{3} \) times higher than that of the standard optimal contest without pre-selection. Formally,
\[
\lim_{n \rightarrow \infty} \frac{S^{(n)}(m^*,n,1)}{S^{(n)}(n,n,1)} = \frac{4}{3}.
\]  
In other words, pre-selection can enhance the highest performance by \( 1/3 \).
\end{theorem}

Complementing the asymptotic results above, we also show that under certain specific ability distributions, such as the uniform distribution, the two-contestant winner-take-all contest remains optimal for any number of participants.

\begin{proposition}[Non-asymptotic Optimal Contest for Uniform Distributions]
For the uniform distribution \(U[0,b]\), the two-contestant winner-take-all contest is optimal for all \( n \geq 2 \).  
\end{proposition}

\subsection{Total Performance}
For total performance objective, we first show that a winner-take-all contest is also optimal for any given shortlist size. 
\begin{proposition}[Optimal Prize Structure for Total Performance] \label{prop:TotalWTA}
    Given the shortlist size \( m \), the prize structure that maximizes total performance is a winner-take-all, i.e., \( V_1 = 1 \) and \( V_l=0\) for all \( l>1\).
\end{proposition}

Then, we analyze how the shortlist size affects total performance when the prize structure is fixed. 
Counterintuitively, if the contest follows a simple contest format, the total performance remains unchanged regardless of the shortlist size. 
More generally, this result extends to arbitrary prize structures, as any such structure can be expressed as a linear combination of simple contests. Formally, 

\begin{proposition}[Independence of Shortlist Size] \label{prop:Allsame}
For any number of contestants \( n \) and ability distribution \( F \), the total performance under any feasible rank-order prize structure \( V_1 \geq \ldots \geq V_m \geq 0 \) with budget constraint \( \sum_{l=1}^{m} V_l \leq 1 \) remains unchanged regardless of the shortlist size \( m \). Specifically, the total performance is given by \( \sum_{l=1}^{m-1} Z_l\mathbb{E}_{X^{(n-l)}}[X] \), where \( Z_l = l(V_l - V_{l+1}) \).   
\end{proposition}

The intuition is that pre-selection increases perceived competition intensity among the admitted contestants, as they face stronger opponents. This compels them to exert more to stay competitive. Notably, under constant marginal costs, the incentive effect of the shortlist on high-ability contestants \emph{exactly offsets} the loss caused by reduced participation.

Having identified the optimal prize structure and the performance invariance across different shortlist sizes, we are now ready to present the optimal contest design.
\begin{theorem}[Optimal Contest for Total Performance]
\label{thm:opt-tp}
For any number of  contestants \( n \) and ability distribution \( F \), the winner-take-all contest with any shortlist size \( m \), (i.e., \(\mathcal{C} = (m, (1, 0, \ldots, 0)) \) where \( 2 \leq m \leq n \)) is optimal for maximizing total performance. The expected total performance achieved is \( \mathbb{E}_{X^{(n-1)}}[X] \).
\end{theorem}

\begin{figure*}[t!]
    \centering
    \includegraphics[width=\textwidth]{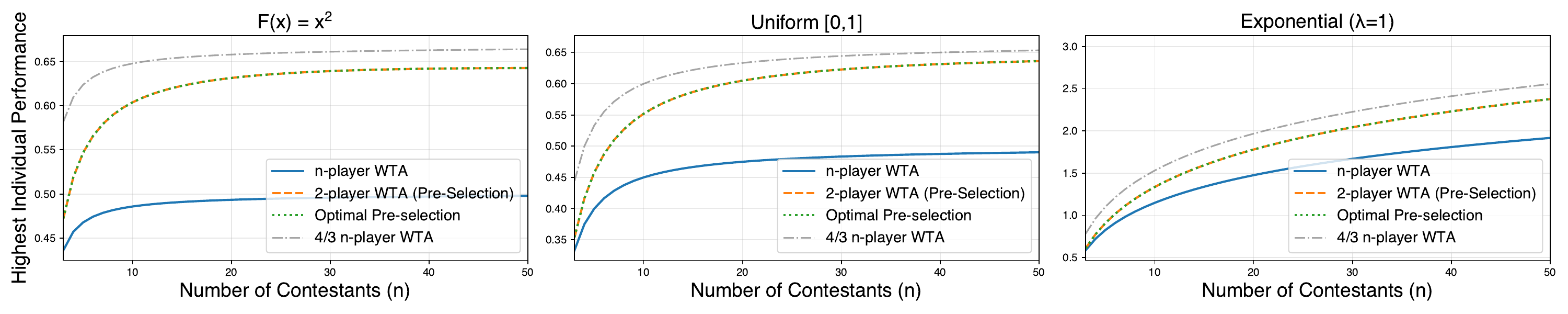}
    \caption{Highest individual performance under optimal pre-selection across three ability distributions. The optimal contest consistently reduces to a 2-contestant WTA contest and rapidly approaches the asymptotic \(4/3\) improvement over the optimal contest without pre-selection.}
    \label{fig:Asymp}
\end{figure*}

\begin{figure*}[t!]
    \centering
    \includegraphics[width=\textwidth]{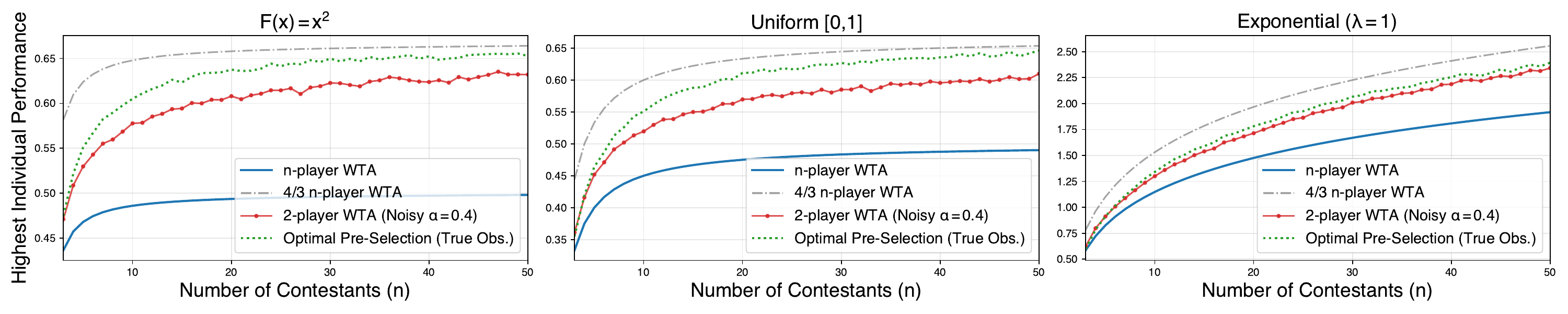}
    \caption{Robustness of pre-selection under noisy ability observations (\(\tilde x_i = (1-\alpha)x_i + \alpha \varepsilon_i\), \(\varepsilon_i \sim F\), \(\alpha=0.4\)). Even when the designer may eliminate high-ability agents due to noise, the induced 2-contestant WTA contest remains close to the ideal benchmark and strictly outperforms the optimal \(n\)-contestant WTA without pre-selection across all tested distributions.}
    \label{fig:Noisy}
\end{figure*}

\vspace{-1em}
\section{Experiments}\label{sec:exp}

Our theoretical results on the optimal contest design under the highest individual performance objective are derived in an asymptotic regime. In this section, we provide numerical evidence under three typical ability distributions (\(F(x)=x^2\), \(\text{U}[0,1]\), and \(\text{Exp}(1)\)), to demonstrate the robustness of these conclusions beyond the asymptotic setting. These distributions represent ability populations with convex, linear, and concave CDFs, covering environments where high-ability agents are dense, uniformly spread, or extremely sparse. The results are shown in Figure~\ref{fig:Asymp} and Figure~\ref{fig:Noisy}.

\textbf{Non-Asymptotic Optimality.}
Although the optimal shortlist size (\(m=2\)) is obtained asymptotically, the numerical results reveal a strikingly stronger phenomenon: for all tested contestant sizes \(n\), the optimal contest with pre-selection consistently reduces to a 2-contestant winner-take-all contest. Moreover, this mechanism approaches the asymptotic \(4/3\) performance improvement at very small \(n\) (\(n \le 8\)), which is stable across all three distributions (It is noteworthy that the apparent universal optimality of the 2-contestant WTA contest in Figure~\ref{fig:Asymp} is \emph{distribution-dependent}, Appendix~\ref{apx:non-optimal} provides a counterexample where larger shortlist sizes become optimal in finite populations.). This indicates that the advantage of shortlist-based contests
persists robustly in finite settings.

\textbf{Robustness to Noisy Observation.}
In practice, designers rarely learn contestants’ exact abilities and must rely on noisy estimates derived from historical ratings.
We model this by adding an independent noise term with the same scale as the ability distribution (\(\tilde x_i = (1-\alpha)x_i + \alpha \varepsilon_i\), \(\varepsilon_i \sim F\)) and set a relatively large noise level (\(\alpha=0.4\)), allowing the designer to mistakenly eliminate high-ability contestants. Even under such substantial noise, the resulting 2-contestant WTA contest still consistently outperforms the optimal \(n\)-contestant WTA without pre-selection, and remains very close to the ideal benchmark where abilities are perfectly observed. This demonstrates that even coarse and noisy ranking information is sufficient for the designer to leverage informational advantages via pre-selection, yielding robust improvements in highest individual performance.


\section{Conclusion and Future Directions}
In this work, we study a rank-order contest with pre-selection, in which a cut-off score must be announced to ensure transparency. We first derive a closed-form characterization of the unique contestant equilibrium performance. Then, we analyze two common designer's objective functions and determine the corresponding optimal contest designs. We also compare the performance obtained by the optimal contest with or without a pre-selection stage, showcasing the power of a shortlist.

Several promising directions remain for future research. First, more general cost functions, such as convex or concave forms, could be incorporated. Second, alternative and more realistic objectives, such as winner performance or threshold-based goals could be explored within our framework. Thirdly, allowing for noisy ability observation is theoretically challenging but interesting. Finally, extending our model to a multi-contest setting presents an interesting avenue, where competition among contest designers offers rich opportunities for both academic inquiry and practical application.

\section*{Acknowledgments and Disclosures}
This work was supported by National Natural Science Foundation of China (No.62472428, No.62602392), Public Computing Cloud, Renmin University of China, the fund for building world-class universities (disciplines) of Renmin University of China, the Natural Science Foundation of Shandong Province (No.ZR2025QC697).

The authors declare no competing interests.

\newpage
\bibliographystyle{abbrvnat}
\bibliography{neurips_2026}


\newpage
\appendix
\setcounter{theorem}{0}
\newtheorem{apptheorem}{Theorem}[section]
\newtheorem{applemma}[apptheorem]{Lemma}

\section{Missing Proofs in Section~\ref{sec:pre}}
\subsection{Proof of Proposition~1}
\begin{proof}
    Let $\mathcal I$ be the set of admitted contestants, and reindex them as $[m]:=\{1,2,\ldots, m\}$. Without loss of generality, assume contestant $i$ is reindexed as contestant $1$. Her posterior belief over the ability values $(x_2, x_3, \ldots, x_m)$ of other contestants is:
    \[
    \begin{aligned}
        & \textstyle \Pr \left[ \bigwedge_{j=2}^m X_j = x_j \mid X_1 = x_1, \mathcal I , x^{(n-m)}\right] 
        \\ = &  \frac{\Pr \left[\bigwedge_{j=1}^m X_j = x_j, \mathcal I , x^{(n-m)}\right]}{\Pr[X_1=x_1, \mathcal I, x^{(n-m)}]}.
    \end{aligned}
    \]

    The denominator can be written as $\Pr[X_1=x_1] \Pr[\mathcal I, x^{(n-m)} \mid X_1 = x_1]$. Here, the event $\mathcal{I} \wedge x^{(n-m)}$ means that all contestants in $\mathcal{I}$ have ability values exceeding $x^{(n-m)}$, some contestant $k \in [n] \backslash \mathcal{I}$ has ability value equal to $x^{(n-m)}$, and all other contestants in $[n] \backslash \mathcal{I}$ have ability values below $x^{(n-m)}$. Since contestant $1$ is confirmed to be admitted, under the condition $X_1 = x_1 > x^{(n-m)}$, this is equivalent to all contestants in $\mathcal{I} \backslash \{1\}$ having ability values above $x^{(n-m)}$, some contestant $k \in [n] \backslash \mathcal{I}$ having ability value equal to $x^{(n-m)}$, and all other contestants in $[n] \backslash \mathcal{I}$ having ability values below $x^{(n-m)}$. Thus,
    \[
    \begin{aligned}
    & \Pr[X_1=x_1, \mathcal I, x^{(n-m)}] \\
    = & \Pr[X_1=x_1] \Pr[\mathcal I, x^{(n-m)} \mid X_1 = x_1] \\
    = & f(x_1) \textstyle \binom{n-m}{1} [X_{m+1}= x^{(n-m)}, \textstyle \bigwedge_{j=2}^m X_j > x^{(n-m)}, \bigwedge_{k=m+2}^n X_k < x^{(n-m)} ] \\
    = & \textstyle (n-m) f(x_1) f(x^{(n-m)}) F(x^{(n-m)})^{n-m-1} \cdot (1-F(x^{(n-m)}))^{m-1}.
\end{aligned}
    \]

    Since the ability values of admitted contestants are no less than those of eliminated contestants, when $(x_2, x_3, \ldots , x_m) > x^{(n-m)} \mathbf{1}$, the numerator becomes:
    \[
    \begin{aligned}
        & \textstyle \Pr \left[\bigwedge_{j=1}^m X_j = x_j, \mathcal I , x^{(n-m)}\right] \\
        = & \textstyle \Pr\left[\bigwedge_{j=1}^m X_j = x_j\right] \Pr\left[\mathcal I , x^{(n-m)} \mid \bigwedge_{j=1}^m X_j = x_j\right] \\
        = & \textstyle \prod_{j=1}^{m} f(x_j) \binom{n-m}{1} \Pr [X_{m+1}=x^{(n-m)}, \\
         & \textstyle \bigwedge_{j=1}^m X_j = x_j > x^{(n-m)}, \bigwedge_{k=m+2}^n X_k <x^{(n-m)} ] \\
        = & \textstyle (n-m) \prod_{j=1}^{m} f(x_j) f(x^{(n-m)}) F(x^{(n-m)})^{n-m-1}.
    \end{aligned}
    \]

    Therefore, the joint posterior probability density can be expressed as:
    \[
    \begin{aligned}
    &p(x_2, \ldots,x_m \mid x_1, x^{(n-m)}) \\
    = &
    \begin{cases} 
    \frac{\prod_{j=2}^{m} f(x_j)}{(1-F(x^{(n-m)}))^{m-1}}, & \text{if } \forall x_j > x^{(n-m)}, \\
    0, & \text{otherwise},
    \end{cases}
    \end{aligned}
    \]
    which is independent of contestant $1$'s own ability value, and thus can be denoted as $p(x_{-1})$.

    For another admitted contestant $k \neq 1$, the marginal probability density function is:
    \[
    \begin{aligned}
         p(x_k) & = \textstyle \int_{x_{-\{1,k\}} > x^{(n-m)}\mathbf{1}} p(x_{-1}) \, dx_{-\{1,k\}} \\
        & = \textstyle \underbrace{\int_{x^{(n-m)}}^{1} \cdots \int_{x^{(n-m)}}^{1}}_{m-2} \frac{\prod_{j=2}^{m} f(x_j)}{(1-F(x^{(n-m)}))^{m-1}} \, \underbrace{dx_2 \dots dx_m}_{m-2} \\
        & = \textstyle \frac{f(x_k)}{1-F(x^{(n-m)})}.
    \end{aligned}
    \]

    Hence, $p(x_1) = \prod_{j=2}^m p(x_j)$, meaning the joint probability density is the product of marginal densities, and the observer's posterior beliefs over any other admitted contestant are mutually independent. Finally, the cumulative distribution function is computed as $P(x_j) = \int_{x^{(n-m)}}^{x_j} p(x) \,dx = \frac{F(x_j) - F(x^{(n-m)})}{1-F(x^{(n-m)})}$. 
\end{proof}

\subsection{Proof of Theorem~1}
\begin{proof}

    We first show that symmetric equilibrium requires monotone strategy function. Fix the cut-off score \(y=x^{(n-m)}\), and suppose all opponents use a common symmetric strategy \(b\). Let
\[
R_b(e)=\sum_{j=1}^m V_j\Pr_j(e)
\]
denote the expected prize obtained by a contestant who chooses performance \(e\), given that all other contestants follow \(b\). Importantly, \(R_b(e)\) is independent of the contestant's own ability \(x\); the ability affects utility only through the cost term \(g(e)/x\).

Suppose, toward a contradiction, that a non-monotone symmetric equilibrium exists. Then there exist two abilities \(x_H>x_L>y\) such that
\[
e_H:=b(x_H)<b(x_L):=e_L .
\]
Since \(b\) is an equilibrium strategy, \(e_H\) is optimal for type \(x_H\), while \(e_L\) is optimal for type \(x_L\). Hence,
\[
R_b(e_H)-\frac{g(e_H)}{x_H}
\ge
R_b(e_L)-\frac{g(e_L)}{x_H},
\]
and
\[
R_b(e_L)-\frac{g(e_L)}{x_L}
\ge
R_b(e_H)-\frac{g(e_H)}{x_L}.
\]
Adding these two inequalities yields
\[
0\le
\left(g(e_H)-g(e_L)\right)
\left(\frac{1}{x_L}-\frac{1}{x_H}\right).
\]
However, since \(e_H<e_L\) and \(g\) is strictly increasing, we have
\(g(e_H)-g(e_L)<0\). Moreover, \(x_H>x_L\) implies
\(\frac{1}{x_L}-\frac{1}{x_H}>0\). Therefore the right-hand side is strictly negative, a contradiction. Hence no symmetric equilibrium can be non-monotone; every symmetric equilibrium strategy must be weakly increasing in ability.

    Next, for analytical convenience, we first introduce some notation. Since the symmetric strategy function \( b: x_i \mapsto e_i \) is increasing, its generalized inverse exists and is denoted by \( \gamma_i(\cdot) \). Here, \( \gamma(e_i) \) represents the ability value \( x_{e_i} \) corresponding to performance \( e_i \) under strategy \( b(\cdot) \), abbreviated as \( \gamma_i \). Thus, for each contestant \( i \), choosing performance \( e_i \) is equivalent to reporting an ability value \( \gamma_i \).

    In the symmetric Bayesian Nash equilibrium, each contestant's performance \( b(x_i) \) is the best response given her own ability \( x_i \), i.e., for all \( i \in [m] \):
    \[
    b(x_i) \in \mathop{\arg \max}_{e_i} u_i := \sum_{l=1}^{m} V_l P_{i,l} - \frac{g(e_i)}{x_i},
    \]
    where \( P_{i,l} \) denotes the probability that contestant \( i \) ranks \( l \)-th when she chooses performance \( e_i \) while all other contestants follow strategy \( b(\cdot) \).

    Since \( e_i \) and \( \gamma_i \) are bijective, the above expression can be rewritten as:
    \[
    b(x_i) \in \left\{ b(\gamma_i) \mid \mathop{\arg \max}_{\gamma_i} \sum_{l=1}^{m} V_l P_{(i,l)} - \frac{g(b(\gamma_i))}{x_i} \right\}.
    \]

    When contestant \( i \) reports ability \( \gamma \), her posterior subjective probability of ranking \( l \)-th is:
    \[
    \begin{aligned}
    & P_{i,l}(\gamma| X_i =x_i) \\
    = & \textstyle \binom{m-1}{l-1} \Pr\left[\bigwedge_{j=1}^{l-1} X_j > \gamma, \bigwedge_{k=l+1}^m X_j < \gamma \mid x_i, \mathcal{I}, x^{(n-m)}\right] \\
    = & \textstyle \binom{m-1}{l-1} \prod_{j=1}^{l-1} \Pr[ X_j > \gamma \mid x_i, \mathcal{I}, x^{(n-m)}] \cdot \textstyle \prod_{k=l+1}^{m} \Pr[X_j < \gamma \mid x_i, \mathcal{I}, x^{(n-m)}] \\
    = & \textstyle \binom{m-1}{l-1} [1-P(\gamma)]^{l-1}P(\gamma)^{m-l}.
\end{aligned}
    \]
    This expression is independent of the contestant's own ability \( x_i \), so we omit the subscripts and conditions and simply denote it as \( P_{l}(\gamma) \).

    Differentiating the utility \( u_i(\gamma) \) with respect to the reported ability \( \gamma \) yields the first-order condition:
    \[
    \textstyle \sum_{l=1}^m P'_l(\gamma)V_l x_i = g'(b(\gamma))b'(\gamma).
    \]

    In equilibrium, each contestant follows the strategy function \( b(\cdot) \), so \( \gamma_i = \gamma(e_i) = \gamma(b(x_i)) = x_i \). Substituting into the first-order condition gives:
    \[
    \textstyle \sum_{l=1}^m P'_l(x_i)V_l x_i = g'(b(x_i))b'(x_i).
    \]

    Since the equilibrium condition holds for any realized ability \( x_i > x^{(n-m)} \), we replace \( x_i \) with variable \( t \) and integrate both sides over \( (x^{(n-m)}, x_i] \) to obtain:
    \[
    \begin{aligned}
        \int_{x^{(n-m)}}^{x_i} \sum_{l=1}^m P'_l(t)V_l t \, dt & = \int_{x^{(n-m)}}^{x_i} g'(b(t))b'(t) \, dt\\
        & = g(b(x_i)).
    \end{aligned}
    \]

    Applying the inverse cost function \( g^{-1}(\cdot) \) to both sides yields the unique candidate expression for the symmetric equilibrium strategy:
    \begin{align}
        b(x ; x^{(n-m)}) = g^{-1}\left(\int_{x^{(n-m)}}^{x}\sum_{l=1}^{m}P'_{l}(t) V_l t \, dt\right). \label{eq:EqRaw} 
    \end{align}

    We now verify that this symmetric strategy indeed constitutes an equilibrium. Substituting the strategy into the utility function and differentiating with respect to \( \gamma \) gives:
    \[
    \begin{aligned}
        u'(\gamma;x_i) & = \sum_{l=1}^m P_l'(\gamma)V_l - \frac{\sum_{l=1}^m P_l'(\gamma)V_l\gamma}{x_i} \\
        & = \sum_{l=1}^m P_l'(\gamma)V_l \left( 1- \frac{\gamma}{x_i}\right).
    \end{aligned}
    \]
    When \( 0 < \gamma < x_i \), \( u' > 0 \); When \( \gamma > x_i \), \( u' < 0 \). Thus, the utility attains its global maximum at \( \gamma = x_i \), meaning that contestant \( i \)'s optimal response is to exhibit performance \( e_i = b(x_i) \) according to the symmetric strategy. This confirms the equilibrium. We now derive the complete expression for \( b(\cdot) \).

    Differentiating \( P_l(\gamma) \) gives:
    \[
    \begin{aligned}
        P_l'(\gamma) = & \textstyle \binom{m-1}{l-1}(m-l)(1-P(\gamma))^{l-1}P(\gamma)^{m-l-1}p(\gamma) 
        \\ & \textstyle - \binom{m-1}{l-1}(l-1)(1-P(\gamma))^{l-2}P(\gamma)^{m-l}p(\gamma).
    \end{aligned}
    \]
    
    Since \( \binom{m-1}{l-1} (m-l) = \frac{(m-1)!}{(l-1)!(m-l-1)!} = \binom{m-1}{l} l \), the first term in \( P'_l(\gamma) \) and the second term in \( P'_{l+1}(\gamma) \) can be combined. Note that the second term in \( P_{1}'(\gamma) \) and the first term in \( P_{m}'(\gamma) \) are zero. The remaining \( 2m -2 \) non-zero terms in \( \sum_{l=1}^{m}P'_{l}(t) V_l \) can be reorganized through telescoping to obtain:
    \[
    \begin{aligned}
    \textstyle \sum_{l=1}^{m}P'_{l}(t) V_l = & \textstyle \sum_{l=1}^{m-1}(V_l-V_{l+1})\binom{m-1}{l-1}(m-l) \textstyle \cdot(1-P(t))^{l-1}P(t)^{m-l-1}p(t).
\end{aligned}
    \]

    Substituting into the equilibrium strategy expression yields the final form:
    \[
    \begin{aligned}
    \textstyle & b(x;x^{(n-m)}) \\
    = & \textstyle g^{-1}(\int_{x^{(n-m)}}^x \sum_{l=1}^{m-1}(V_l-V_{l+1})\binom{m-1}{l-1}(m-l) \cdot(1-P(t))^{l-1}P(t)^{m-l-1}p(t)t\, dt).
\end{aligned}
    \]

    Finally, we verify the strategy's validity. Since we consider a rank-order reward structure where \( V_l > V_{l+1} \), the integrand is positive almost everywhere, ensuring the non-negativity and monotonicity of \( b(x) \). Moreover, letting \( T \) denote the integrand in \( b(\cdot) \), since \( T \) is bounded, we have:
    \[
    \lim_{x \to x^{(n-m)}} g(b(x;x^{(n-m)})) \to 0.
    \]
    Thus, we have established the existence of a unique Bayesian Nash equilibrium, completing the proof.

    The equilibrium performance expression shows that contestants are incentivized by the differences between consecutive prizes, $(V_{l} - V_{l+1})$, rather than by their absolute values. Furthermore, when the cost function $g(\cdot)$ is linear, the marginal impact of each prize gap on performance is additive, indicating these effects are independent of one another.
\end{proof}

\section{Missing Proofs in Section~\ref{sec:opt}}

\subsection{Proof of Corollary~1}
\begin{proof}
    From Equation~\eqref{eq:EqRaw}, the contribution of the prize \( V_m \) to the performance of a participant with ability \( x \) is \( \int_0^x P'_m(t)V_m t \, dt \), where:
    \[
    \begin{aligned}
    & P_m'(t) =  \textstyle \binom{m-1}{m-1}(m-m)(1-P(t))^{m-1}P(t)^{m-m-1} \textstyle - \binom{m-1}{m-1}(m-1)(1-P(t))^{m-2}P(t)^{m-m} \\
    & =  - (m-1)(1-P(t))^{m-2} \leq 0,
\end{aligned}
    \]
    holds for all \( t > 0 \), and \( P'_m(t) < 0 \) on a subset of \( (0,x] \) with non-zero measure.
    
    Since \( g^{-1}(\cdot) \) is strictly monotonically increasing, for any realization of the eliminated ability \( x^{(n-m)} < x \), we have:
    \[
    \begin{aligned}
        b(x ; x^{(n-m)}) & = g^{-1}\left(\int_{x^{(n-m)}}^{x}\sum_{l=1}^{m}P'_{l}(t) V_l t \, dt\right) \\
        & < g^{-1}\left(\int_{x^{(n-m)}}^{x}\sum_{l=1}^{m-1}P'_{l}(t) V_l t \, dt\right),
    \end{aligned}
    \]
    which means setting \( V_m = 0 \) improves the performance of every admitted participant with any ability level. This completes the proof.
\end{proof}

\subsection{Proof of Proposition~2}
\begin{proof}
    We first consider the ex-post utility. The designer's utility is given by:
    \[
    \begin{aligned}
        & u_d(\vec{x}) = \textstyle \vec{c} \cdot e(\vec{x}) 
        \\ & = \textstyle\sum_{i=1}^m c_i b(x^{(n-i+1)};x^{(n-m)})
        \\ & = \textstyle\sum_{i=1}^m c_i \left ( k^{-1} \int_{x^{(n-m)}}^{x^{(n-i+1)}}\sum_{l=1}^{m}P'_{l}(t) V_l t \, dt \right ) 
        \\ & = \textstyle \sum_{i=1}^m c_i \bigg ( k^{-1} \int_{x^{(n-m)}}^{x^{(n-i-1)}} \sum_{l=1}^{m-1}(V_l-V_{l+1})\binom{m-1}{l-1}
         \cdot (m-l)(1-P(t))^{l-1}P(t)^{m-l-1}p(t)t\, dt \bigg) 
        \\ & = \sum_{l=1}^{m-1} (V_l-V_{l+1}) \bigg ( k^{-1}\sum_{i=1}^{m}c_i \int_{x^{(n-m)}}^{x^{(n-i-1)}} \binom{m-1}{l-1}
         \cdot (m-l)(1-P(t))^{l-1}P(t)^{m-l-1}p(t)t\, dt \bigg ),
    \end{aligned}
    \]
    For convenience, we define the term independent of the prize structure \( \vec{V} \) as \( T_l > 0 \). The designer's optimization problem becomes:
    \[
    \begin{aligned}
        \mathop{\arg \max}_{m,\vec{V}} \quad & \sum_{l=1}^{m-1}(V_l-V_{l+1}) \, T_l \\
        \text{s.t.} \quad & \sum_{l=1}^{m} V_l \leq B \\
            & V_l  \geq V_{l+1} \geq 0.
    \end{aligned}
    \]

    Let \( Z_l = l(V_l-V_{l+1}) \). From Corollary~\ref{coro:ZeroPrize}, we know that the optimal solution must satisfy \( V_m = 0 \). Thus, the optimization problem can be rewritten as:
    \[
    \begin{aligned}
        \mathop{\arg \max}_{m,\vec{V}} \quad & \sum_{l=1}^{m-1}Z_l \, \frac{T_l}{l}\\
        \text{s.t.} \quad & \sum_{l=1}^{m} Z_l \leq B \\
            & Z_l \geq 0,
    \end{aligned}
    \]
    where it is straightforward to verify that \( \sum_{l=1}^{m-1} Z_l = \sum_{l=1}^{m} V_l -mV_m\leq B-0= B \).

    For a fixed number of winners \( m \), the optimal solution corresponds to setting \( Z_{l^*} = B \) for some \( l^* \) and \( Z_l = 0 \) for all others. This implies that the first \( l^* \) prizes are equal to \( \frac{B}{l^*} \), while all subsequent prizes are \( 0 \), i.e., \( V_1 = \ldots=V_{l^*} = \frac{B}{l^*} > V_{l^*+1} = \ldots = V_m = 0 \), with \( l^* < m \). This corresponds to a non-trivial simple contest. By enumerating all possible \( m \in [n-1] \), the global optimum must also be a non-trivial simple contest, as desired.

    For the ex-ante utility \( u_d = \mathbb{E}[ \vec{c} \cdot e(\vec{x})] \), due to the linearity of expectation, the linearity of the cost function, and the linearity of the performance function in \( (V_l - V_{l+1}) \), the original optimization objective can still be transformed into a positive linear combination of \( Z_l \). Therefore, the same analysis applies, and the optimal contest design must be a non-trivial simple contest. This completes the proof. 
\end{proof}

\subsection{Proof of Lemma~1}
\begin{proof}
We first derive the joint probability density function of the conditional distribution. The joint density of the order statistics \(X^{(1)}, X^{(2)}, \ldots, X^{(n)}\) is:
\[
\begin{aligned}
    f_{X^{(1)}, \ldots, X^{(n)}}(x_1, \ldots, x_n) = & n! \prod_{i=1}^n f(x_i), \quad x_1 < x_2 < \cdots < x_n.
\end{aligned}
\]

Consider the joint density of the subset \((X^{(r)}, X^{(r+1)}, \ldots, X^{(n)})\). To obtain the joint density of \(X^{(r)}, \ldots, X^{(n)}\), we marginalize over the first \(r-1\) order statistics \(X^{(1)}, \ldots, X^{(r-1)}\) by integrating:
\[
\begin{aligned}
    f_{X^{(r)}, \ldots, X^{(n)}}(x_r, \ldots, x_n) = & \textstyle \int_{x_1 < \cdots < x_{r-1} < x_r} f_{X^{(1)}, \ldots, X^{(n)}}(x_1, \ldots, x_n) \, dx_1 \cdots dx_{r-1}.
\end{aligned}
\]

Substituting the joint density of the full order statistics:
\[
\textstyle = n! \prod_{i=r}^n f(x_i) \cdot \int_{x_1 < \cdots < x_{r-1} < x_r} \prod_{i=1}^{r-1} f(x_i) \, dx_1 \cdots dx_{r-1}.
\]

We need to compute:
\[
\int \cdots \int_{x_1 < \cdots < x_{r-1} < x_r} \prod_{i=1}^{r-1} f(x_i) \, dx_1 \cdots dx_{r-1}.
\]

This integral represents the joint density of \(x_1, \ldots, x_{r-1}\) under the constraint \(x_1 < \cdots < x_{r-1} < x_r\). It is equivalent to integrating the joint density of the order statistics of \(r-1\) iid random variables over \((-\infty, x_r)\).

For \(r-1\) i.i.d. random variables \(Z_1, \ldots, Z_{r-1}\), the joint density of their order statistics \(Z^{(1)} < \cdots < Z^{(r-1)}\) is:
\[
\begin{aligned}
    f_{Z^{(1)}, \ldots, Z^{(r-1)}}(z_1, \ldots, z_{r-1}) = & (r-1)! \prod_{i=1}^{r-1} f(z_i), \quad z_1 < \cdots < z_{r-1}.
\end{aligned}
\]

Thus, the marginal density of \(Z^{(r-1)}\) can be obtained by integration:
\[
\begin{aligned}
    f_{Z^{(r-1)}}(z_{r-1}) = & \int_{z_1 < \cdots < z_{r-2} < z_{r-1}} (r-1)! \prod_{i=1}^{r-1} f(z_i) \, dz_1 \cdots dz_{r-2}.
\end{aligned}
\]

However, we need the integral over all \(z_1, \ldots, z_{r-1}\) in \((-\infty, x_r)\):
\[
\begin{aligned}
    & \int \cdots \int_{z_1 < \cdots < z_{r-1} < x_r} \prod_{i=1}^{r-1} f(z_i) \, dz_1 \cdots dz_{r-1} 
    \\ = & \frac{1}{(r-1)!} P(Z^{(r-1)} < x_r).
\end{aligned}
\]

Since \(Z^{(r-1)}\) is the maximum of \(r-1\) i.i.d. random variables, its distribution function is:
\[
P(Z^{(r-1)} < x_r) = [F(x_r)]^{r-1}.
\]

Therefore,
\[
\int_{x_1 < \cdots < x_{r-1} < x_r} \prod_{i=1}^{r-1} f(x_i) \, dx_1 \cdots dx_{r-1} = \frac{[F(x_r)]^{r-1}}{(r-1)!}.
\]

Substituting this result back, we obtain the joint density of the subset:
\[
\begin{aligned}
    f_{X^{(r)}, \ldots, X^{(n)}}(x_r, \ldots, x_n) = & n! \prod_{i=r}^n f(x_i) \cdot \frac{[F(x_r)]^{r-1}}{(r-1)!} \\
    = & \frac{n!}{(r-1)!} [F(x_r)]^{r-1} \prod_{i=r}^n f(x_i).
\end{aligned}
\]

The marginal density of the \(r\)-th order statistic \(X^{(r)}\) is:
\[
f_{X^{(r)}}(x) = \frac{n!}{(r-1)!(n-r)!} [F(x)]^{r-1} [1 - F(x)]^{n-r} f(x).
\]

Given \(X^{(r)} = a\), the conditional density is:
\[
\begin{aligned}
    f_{X^{(r+1)}, \ldots, X^{(n)} \mid X^{(r)}}(y_{r+1}, \ldots, y_n \mid a) = & \frac{f_{X^{(r)}, X^{(r+1)}, \ldots, X^{(n)}}(a, y_{r+1}, \ldots, y_n)}{f_{X^{(r)}}(a)}.
\end{aligned}
\]

Substituting the joint and marginal densities, we compute the ratio:
\[
\begin{aligned}
    f_{X^{(r+1)}, \ldots, X^{(n)} \mid X^{(r)}}(y_{r+1}, \ldots, y_n \mid a) = &  \frac{\frac{n!}{(r-1)!} [F(a)]^{r-1} f(a) \prod_{i=r+1}^n f(y_i)}{\frac{n!}{(r-1)!(n-r)!} [F(a)]^{r-1} [1 - F(a)]^{n-r} f(a)} \\
    = & \frac{\prod_{i=r+1}^n f(y_i)}{\frac{1}{(n-r)!} [1 - F(a)]^{n-r}} \\
    = & (n-r)! \frac{\prod_{i=r+1}^n f(y_i)}{[1 - F(a)]^{n-r}}.
\end{aligned}
\]

Let \(m = n - r\) and set \(y_j^* = y_{r+j}\) for \(j = 1, 2, \ldots, m\). Then:
\[
\begin{aligned}
    f_{X^{(r+1)}, \ldots, X^{(n)} \mid X^{(r)}}(y_1^*, \ldots, y_m^* \mid a) = & m! \prod_{j=1}^m \frac{f(y_j^*)}{1 - F(a)}, \quad a < y_1^* < y_2^* < \cdots < y_m^*.
\end{aligned}
\]

Next, under the truncated distribution, the joint density of the order statistics \(Y^{(1)} < \cdots < Y^{(m)}\) is:
\[
\begin{aligned}
    f_{Y^{(1)}, \ldots, Y^{(m)}}(y_1, \ldots, y_m) = &  m! \prod_{i=1}^m f_Y(y_i) = m! \prod_{i=1}^m \frac{f(y_i)}{1 - F(a)}, \quad a < y_1 < \cdots < y_m.
\end{aligned}
\]

We observe that the conditional joint density \(\left(X^{(r+1)}, \ldots, X^{(n)}\right) \mid X^{(r)} = a\) is identical to the joint density of \(\left(Y^{(1)}, \ldots, Y^{(m)}\right)\). Therefore, \(\left(X^{(n-m+1)}, \ldots, X^{(n)}\right) \mid X^{(n-m)} = a\) follows the same distribution as \(\left(Y^{(1)}, \ldots, Y^{(m)}\right)\), completing the proof.
\end{proof}

\subsection{Proof of Lemma~2}
\begin{proof}
For \(\mathbb{E}_{x \sim X^{(n)} | X^{(n-m)}=y}[b(x;y)]\), we now have:
\[
\begin{aligned}
    & \int_0^\infty \bigg ( \int_0^x \sum_{l=1}^{m-1}Z_l\binom{m-1}{l}(1-P(t))^{l-1} \cdot P(t)^{m-l-1}p(t)t\, dt \bigg ) dP_y^{(m)}(x) \\
    = & \int_0^\infty \sum_{l=1}^{m-1}Z_l\binom{m-1}{l}(1-P(t))^{l-1}P(t)^{m-l-1} \cdot p(t)t \bigg ( \int_t^\infty dP_y^{(m)}(x) \bigg ) dt \\
    = & \sum_{l=1}^{m-1}Z_l \int_0^\infty [1-P(t)^{m}]\binom{m-1}{l}(1-P(t))^{l-1} \cdot P(t)^{m-l-1}p(t)t \, dt.
\end{aligned}
\]
    The optimal reward structure is still a Simple Contest, i.e., some \( Z_l = 1 \), corresponding to the expression:
\[ 
\int_0^\infty [1-P(t)^{m}]\binom{m-1}{l}(1-P(t))^{l-1}P(t)^{m-l-1}p(t)t \, dt.
\] 
Transforming this expression into quantile form via the substitution \( q = 1 - P(t) \). Since \( f(t) > 0 \), there is a bijective mapping between \( q \) and \( t \), denoted as \( t = x(q) = P^{-1}(1-q) \). Then, \( dt  = -\frac{1}{P'(P^{-1}(1-q))} =-\frac{1}{p(t)} \, dq\). The expression becomes:
\[
\int_0^1 \binom{m-1}{l} [1-(1-q)^m]q^{l-1}(1-q)^{m-l-1}x(q)\, dq. 
\]
Let \( g_l(q) = \binom{m-1}{l} [1-(1-q)^m]q^{l-1}(1-q)^{m-l-1} \). Then the expression becomes:
\[
\begin{aligned}
    \int_0^1 g_l(q) v(q) \,dq = &  \int_0^1 g_l(q) \left ( \int_q^1 -x'(t) \,dt  \right ) \,dq \\
    = & \int_0^1 g_l(q) \left ( \int_q^1 |x'(t)| \,dt  \right ) \,dq \\
    = & \int_0^1 |x'(q)| \left ( \int_0^q g_l(t)  \,dt  \right ) \,dq.
\end{aligned}
\]
Let the inner integral be \( G_{l}(q)=\int_0^q g_l(t)\, dt \). Then the expression becomes \( \int_0^1 |x'(q)| G_l(q) \, dq \), where \( G_l(q) \) is distribution-independent.
\end{proof}

\subsection{Proof of Lemma~3}
\begin{proof}
    Let \( F(x) \) and \( G(x) \) denote the antiderivatives of \( f(x) \) and \( g(x) \), respectively. Then, \( (F(x)/G(x))' < 0 \) is equivalent to
    \[
    f(x)G(x) < g(x)F(x),
    \]
    which can be rewritten as
    \[
    \frac{f(x)}{g(x)} < \frac{F(x)}{G(x)}.
    \]
    Since
    \[
    \begin{aligned}
        F(x) & = \int_0^x f(t)\, dt = \int_0^x \frac{f(t)}{g(t)} g(t) \, dt \\ & > \frac{f(x)}{g(x)} \int_0^x g(t)\, dt = \frac{f(x)}{g(x)} G(x),
    \end{aligned}
    \]
    we obtain
    \[
    \frac{F(x)}{G(x)} > \frac{f(x)}{g(x)}.
    \]
    Consequently,
    \[
    \left ( \frac{\int_0^x f(t) \, dt}{ \int_0^x g(t) \, dt} \right )' < 0.
    \]
\end{proof}

\subsection{Proof of Lemma~4}
\begin{proof}
We can calculate that
    \[
    \begin{aligned}
    \int_0^{x}t^{a-1}(x-t)^{b-1}\, dt = & x^{a+b-1} \int_0^x (\frac{t}{x})^{a-1}(\frac{t}{x})^{b-1}(x^{-1} \, dt ) \\
    = & x^{a+b-1} \int_{0}^{1}t^{a-1}(1-t)^{b-1}\, dt
\end{aligned}
    \]
    Recall that for beta function $B(a,b)$, it holds that:
    \[
    B(a,b) = \int_{0}^{1}t^{a-1}(1-t)^{b-1}\, dt = \frac{\Gamma(a)\Gamma(b)}{\Gamma(a+b)}.
    \]
    This gives:
    \[
    \begin{aligned}
    \int_0^{x}t^{a-1}(x-t)^{b-1}\, dt & = \frac{\Gamma(a)\Gamma(b)}{\Gamma(a+b)} x^{a+b-1} \\
    & = \frac{(a-1)!(b-1)!}{(a+b-1)!}x^{a+b-1},
\end{aligned}
    \]
    where last equality holds since $\Gamma(m) = (m-1)!$ for any $m\in \mathbb{N}$, which completes the proof.
\end{proof}

\subsection{Proof of Proposition~3}
\begin{proof}
    Given the number of winners \( m \) and the cut-off score \( y \), the highest individual performance is expressed as \(\mathbb{E}_{x \sim X^{(n)} | X^{(n-m)}=y}[b(x;y)]\). According to Lemma~\ref{lem:Truncate}, this is equivalent to taking the expectation of the order statistic \( P_y^{(m)} \) of the truncated distribution \( P_y \), i.e., \(\mathbb{E}_{x \sim P_y^{(m)}}[b(x;y)]\). 

According to Lemma~\ref{lem:iso1}, the coefficient of $Z_l$ is:
\[
\textstyle \int_0^1 |x'(q)| \left( \int_0^q \binom{m-1}{l} [1-(1-t)^m]t^{l-1}(1-t)^{m-l-1} \,dt \right) \,dq.
\]

To prove that winner-take-all is optimal for any distribution, since \( |x'(q)| \) can take any positive real value, it is required that \(\forall q \in [0,1]\, , l \neq 1\), \( G_{l}(q) \leq G_{1}(q) \) holds.

We first prove that \( G_{l}(q)/G_{l+1}(q) \) is monotonically decreasing in \( q \). The ratio of the integrands is:
\[
\begin{aligned}
    \frac{g_l(t)}{g_{l+1}(t)} & = \frac{\binom{m-1}{l}[1-(1-t)^m]t^{l-1}(1-t)^{m-l-1}}{\binom{m-1}{l+1}[1-(1-t)^m]t^{l}(1-t)^{m-l-2}} 
    \\ & = \frac{l+1}{m-l-1} \cdot \frac{1-t}{t}.  
\end{aligned}
\] 
This ratio is monotonically decreasing in \( t \). By Lemma~\ref{lem:IntDerMon}, \( G_{l}(q)/G_{l+1}(q) \) is monotonically decreasing in \( q \). 
    
To prove \( G_{l}(q) \geq G_{l+1}(q) \) holds for all \( q \in [0,1] \), it suffices to show \( G_{l}(1) \geq G_{l+1}(1) \). Expanding \( G_l(1) \):
$$
\begin{aligned}
    & \int_0^1 \binom{m-1}{l}[1-(1-t)^m]t^{l-1}(1-t)^{m-l-1} \, dt \\
    = & \int_0^1 \binom{m-1}{l}t^{l-1}(1-t)^{m-l-1} \, dt - \int_0^1 \binom{m-1}{l}t^{l-1}(1-t)^{2m-l-1} \, dt .
\end{aligned}
$$

Using Lemma~\ref{lem:betaInt}, the integrals can be expressed as fractions:
$$
\begin{aligned}
    = & \binom{m-1}{l} (B(l, m-l) - B(l, 2m-l)) \\
    = & \frac{(m-1)!}{(m-l-1)!\,l!} \cdot\left ( \frac{(l-1)!(m-l-1)!}{(m-1)!} - \frac{(l-1)!(2m-l-1)!}{(2m-1)!} \right ) \\
    = & \frac{1}{l}\left ( 1 - \frac{(m-1)!(2m-l-1)!}{(m-l-1)!(2m-1)!} \right ) .
\end{aligned}
$$

Next, we prove \( G_l(1) - G_{l+1}(1) \geq 0 \) for all \( l \in [m-1] \).
After simplifying \( G_l(1)-G_{l+1}(1) \), we obtain:
\[
\begin{aligned}
    (l+1) \bigg [ (m-l-1)!(2m-1)! &- (m-1)!(2m-l-1)! \bigg ] \\
    - & l \bigg [ (m-l-1)!(2m-1)! -(m-l-1)(m-1)!(2m-l-2)! \bigg ].
\end{aligned}
\] 
Expanding and combining like terms gives:
\[
\begin{aligned}
    (m-l-1)!(2m-1)! - & (l+1)(m-1)!(2m-l-1)! \\
    + &l(m-l-1)(m-1)!(2m-l-2)!.
\end{aligned}
\]
Further simplifying the last two terms:
\[
\begin{aligned}
    (m-1)!(2m-l-2)! & \cdot [l(m-l-1)]-(l+1)(2m-l-1)] \\
    = &-(m-1)!(2m-l-2)![ml+(2m-l-1)] \\
    = &-l(m!)(2m-l-2)!-(m-1)!(2m-l-1)!.
\end{aligned}
\]
Now the entire expression becomes:
\[
\begin{aligned}
    & (m-l-1)!(2m-1)! -l(m!)(2m-l-2)!-(m-1)!(2m-l-1)! \\
    & =  (m-l-1)!(2m-l-2)! \bigg [ \prod_{k=0}^l (2m-k-1) -l\prod_{k=0}^l (m-k) - (2m-l-1) \prod_{k=1}^l (m-k)\bigg ] \\
    & =  (m-l-1)!(2m-l-2)! \left [ \prod_{k=0}^l (2m-k-1) - ( l + \frac{2m-l-1}{m} ) \prod_{k=0}^l (m-k) \right ] \\
    & =  (m-l-1)!(2m-l-2)! \left [ \prod_{k=0}^l (2m-k-1) - ( l\frac{m-1}{m}+2-\frac{1}{m} ) \prod_{k=0}^l (m-k) \right ] \\
    & \geq  (m-l-1)!(2m-l-2)! \left [ \prod_{k=0}^l (2m-k-1) - ( l+2 ) \prod_{k=0}^l (m-k) \right ].
\end{aligned}
\]
To prove \( G_{l}(1) \geq G_{l+1}(1) \), it suffices to show \( \prod_{k=0}^l (2m-k-1) - ( l+2 ) \prod_{k=0}^l (m-k) \geq 0 \), i.e., \( P(l) = \prod_{k=0}^l \frac{2m-k-1}{m-k} = \prod_{k=0}^l \left(1+\frac{m-1}{m-k}\right) \geq l+2 \). 

We prove this by induction:

\underline{Base case:} For \( l=1 \), \( P(1) = \frac{(2m-1)(2m-2)}{m(m-1)} = \frac{2(2m-1)}{m} \geq 3 \) holds. 

\underline{Inductive step:} Assume \( P(l) \geq l+2 \) holds. We show \( P(l+1) \geq l+3 \). By definition, \( P(l+1) = P(l) \frac{2m-l-2}{m-l-1} = P(l) \left ( 2 + \frac{l}{m-l-1}\right ) \geq 2 P(l) \geq 2l+4 \geq l+3 \). 

By induction, \( \prod_{k=0}^l (2m-k-1) - ( l+2 ) \prod_{k=0}^l (m-k) \geq 0 \) holds. 

In conclusion, we have shown \( G_l(1) \geq G_{l+1}(1) \). By the monotonicity of the ratio, \( G_l(q) \geq G_{l+1}(q) \) holds for all \( q \in [0,1] \). Therefore, for any distribution \( |x'(q)| > 0 \) and \( l \neq 1 \), \( \int_0^1 |x'(q)| G_1(q) \, dq \geq \int_0^1 |x'(q)| G_l(q) \, dq \) holds. Thus, winner-take-all is always the optimal reward structure.
\end{proof}

\subsection{Proof of Lemma~5}
\begin{proof}
The raw expression is:
\[
\begin{aligned}
    \int_0^\infty \left ( \int_y^\infty x (m-1) [1-P(x)^m] P(x)^{m-2} p(x) \, dx \right ) \\ 
    \cdot n \binom{n-1}{m} F(y)^{n-m-1}(1-F(y))^m f(y) \, dy.
\end{aligned}
\]

Let \( u = F(y) \), \( v = F(x) \), then \( dv = f(x)\, dx\). Since \( P(x) = \frac{F(x)-F(y)}{1-F(y)} = \frac{v-u}{1-u} \) and \( p(x) = \frac{f(x)}{1-u} \), the expression transforms to:
\[
\begin{aligned}
    & \int_0^1 \bigg ( \int_{u}^1 x(v)(m-1) [1-(\frac{v-u}{1-u})^m] (\frac{v-u}{1-u})^{m-2} \\
    & \cdot (\frac{1}{1-u}) \, dv \bigg ) n \binom{n-1}{m} u^{n-m-1}(1-u)^m \, du \\
    = & \int_0^1 \bigg ( \int_{u}^1 x(v)n (m-1) \binom{n-1}{m}[1-(\frac{v-u}{1-u})^m] \cdot (v-u)^{m-2} \, dv \bigg ) u^{n-m-1}(1-u) \, du.
\end{aligned}
\]

Next, we swap the order of integration and separate the distribution-independent terms:
\[
\begin{aligned}
    & \int_0^1 \bigg ( \int_0^v n (m-1) \binom{n-1}{m}[1-(\frac{v-u}{1-u})^m](v-u)^{m-2}
     \cdot u^{n-m-1}(1-u) \, du \bigg ) x(v) \, dv \\
= & \int_0^1 x'(w) \bigg ( \int_w^1 \int_0^v n (m-1) \binom{n-1}{m} [1-(\frac{v-u}{1-u})^m] 
\\ & \cdot (v-u)^{m-2} u^{n-m-1} (1-u) \, du \, dv \bigg ) \, dw.
\end{aligned}
\]
\end{proof}

\subsection{Proof of Lemma~6}
\begin{proof}
We now simplify the expression for \( H_m(w) \). For convenience, we temporarily omit the coefficient \( n (m-1) \binom{n-1}{m} \). Performing another change of variables, let \( z = \frac{v-u}{1-u} \), transforming from \( (u, v) \) to \( (u, z) \). The Jacobian matrix is:
\[
\begin{aligned}
    J & = \frac{\partial(u,v)}{\partial(u,z)} \\
    & = \begin{bmatrix}
        1 & 0 \\ \frac{\partial u}{\partial z} & 1-u
    \end{bmatrix}.
\end{aligned}
\]

Thus, \( \det J = 1 - u \). By the change of variables formula for multiple integrals \( \iint f(x, y) \, dx dy = \iint f(x(u,v), y(u,v)) |\det J|\, du \, dv \), we obtain:
\[
\int_w^1 \int_0^v [1 - z^m] z^{m-2} u^{n-m-1} (1 - u)^m \, du \, dz.
\]

Consider the transformation of the integration region. The original domain is \( u \in [0,1], v \in [\max(u,w), 1] \), which is depicted in Figure~\ref{fig:OriRegion}. Substituting \( z \), we get \( z \in [\frac{\max(u,w) - u}{1 - u}, 1] \). 
\begin{figure}[ht]
    \centering
    \begin{tikzpicture}[>=latex,scale=3]

\draw[->] (0,0) -- (1.1,0) node[right] {$v$}; 
\draw[->] (0,0) -- (0,1.1) node[above] {$u$}; 

\def\w{0.25}

\draw (\w, -0.02) -- (\w, 0.02) node[below] {$w$};
\draw (1, -0.02) -- (1, 0.02) node[below] {$1$};

\draw (-0.02, \w) -- (0.02, \w) node[left] {$w$};
\draw (-0.02, 1) -- (0.02, 1) node[left] {$1$};

\draw[dashed] (\w, 0) -- (\w, \w) node[below right] {$v = w$};

\fill[gray!30] (\w, 0) -- (\w, \w) -- (1, 1) -- (1, 0) -- cycle;

\draw[thick] (\w, \w) -- (1, 1) node[above right] {$u = v$};


\end{tikzpicture}
    \caption{Original integration region of \((u,v)\).}
    \label{fig:OriRegion}
\end{figure}
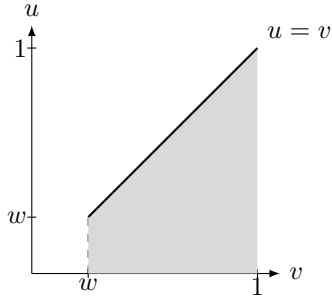

After swapping the integration order, the region splits into two parts: \( u \in [0,w], z \in [\frac{w - u}{1 - u}, 1] \) (The dark gray region in Figure~\ref{fig:Region2}) and \( u \in [w,1], z \in [0,1] \) (The light gray region in Figure~\ref{fig:Region2}). 
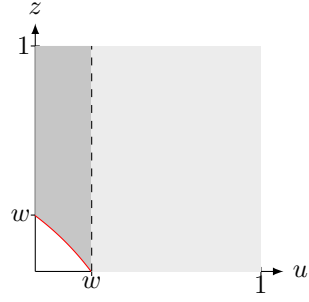
\begin{figure}[ht]
    \centering
    \begin{tikzpicture}[>=latex,scale=3]

\def\w{0.25}

\draw[->] (0,0) -- (1.1,0) node[right] {$u$}; 
\draw[->] (0,0) -- (0,1.1) node[above] {$z$}; 

\draw (\w, -0.02) -- (\w, 0.02) node[below] {$w$};
\draw (1, -0.02) -- (1, 0.02) node[below] {$1$};

\draw (-0.02, \w) -- (0.02, \w) node[left] {$w$};
\draw (-0.02, 1) -- (0.02, 1) node[left] {$1$};

\draw[thick, red] plot[domain=0:\w, samples=50] (\x, {(\w-\x)/(1-\x)});

\fill[darkgray!30, domain=0:\w, samples=50] 
     plot(\x, {(\w-\x)/(1-\x)}) -- (\w,0) -- (\w,1) -- (0, 1) -- cycle ;

\fill[lightgray!30] (\w,0) rectangle (1,1);

\draw[dashed] (\w, 0) -- (\w, 1) ; 

\end{tikzpicture}
    \caption{Integration region of \((u,z)\).}
    \label{fig:Region2}
\end{figure}

The integral over the second region is:
\[
\begin{aligned}
    & \int_w^1 \int_0^1 [1 - z^m] z^{m-2} u^{n-m-1} (1 - u)^m \, du \, dz \\
    = & \left( \int_0^1 [1 - z^m] z^{m-2} \, dz \right ) \left( \int_w^1 u^{n-m-1} (1 - u)^m \, du \right ).
\end{aligned}
\]

To simplify the integral over the first region, we again swap the integration order. The original region is \( u \in [0,w], z \in [0,1] \) with \( z \geq \frac{w - u}{1 - u} \). 

After swapping, it splits into \( z \in [0,w], u \in [\frac{w - z}{1 - z}, w] \) (The black region in Figure~\ref{fig:Region3}) and \( z \in [w,1], u \in [0,w] \) (The dark gray region in Figure~\ref{fig:Region3}).
\begin{figure}[ht]
    \centering
    \begin{center} 
\begin{tikzpicture}[>=latex,scale=3]

\def\w{0.25}

\draw[->] (0,0) -- (1.1,0) node[right] {$u$}; 
\draw[->] (0,0) -- (0,1.1) node[above] {$z$}; 

\draw (\w, -0.02) -- (\w, 0.02) node[below] {$w$};
\draw (1, -0.02) -- (1, 0.02) node[below] {$1$};

\draw (-0.02, \w) -- (0.02, \w) node[left] {$w$};
\draw (-0.02, 1) -- (0.02, 1) node[left] {$1$};

\draw[thick, red] plot[domain=0:\w, samples=50] (\x, {(\w-\x)/(1-\x)});

\fill[black!70, domain=0:\w, samples=50] 
     plot(\x, {(\w-\x)/(1-\x)}) -- (\w,0) -- (\w,\w) -- (0, \w) -- cycle ;

\fill[darkgray!30] (0,\w) rectangle (\w,1);

\draw[dashed] (0, \w) -- (\w, \w) ; 

\end{tikzpicture}
\end{center}
    \caption{Integration region of \((u,z)\) when \(u \in [0,w]\).}
    \label{fig:Region3}
\end{figure}
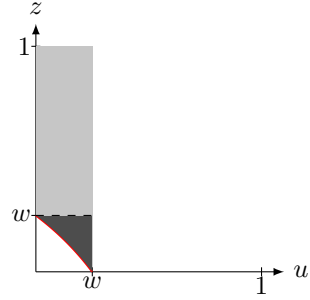

Thus:
\[
\begin{aligned}
    & \int_0^w \int_{\frac{w - u}{1 - u}}^1 [1 - z^m] z^{m-2} u^{n-m-1} (1 - u)^m \, du \, dz =  \int_w^1 \int_0^w \cdots \, dz \, du + \int_0^w \int_{\frac{w - z}{1 - z}}^w \cdots \, dz \, du \\
    & =  \left( \int_w^1 [1 - z^m] z^{m-2} \, dz \right ) \left( \int_0^w u^{n-m-1} (1 - u)^m \, du \right ) \\
    & + \int_0^w \left ( \int_{\frac{w - z}{1 - z}}^w [1 - z^m] z^{m-2} \, dz \right ) u^{n-m-1} (1 - u)^m \, du.
\end{aligned}
\]

Note that:
\[
\begin{aligned}
    \int_a^b [1 - z^m] z^{m-2} \, dz = & \int_a^b z^{m-2} \, dz - \int_a^b z^{2m - 2} \, dz \\
    = & \frac{1}{m - 1} z^{m - 1} \big \vert_a^b - \frac{1}{2m - 1} z^{2m - 1} \big \vert_a^b.
\end{aligned}
\]
When \( a = 0, b = 1 \), this evaluates to \( \frac{1}{m - 1} - \frac{1}{2m - 1} \), and when \( a = w, b = 1 \), it becomes \( \frac{1 - w^{m - 1}}{m - 1} - \frac{1 - w^{2m - 1}}{2m - 1} \).

Let \( B_w(a+1, b+1) \) denote the incomplete beta function \( \int_0^w t^{a-1} (1 - t)^{b-1} \, dt \), where the complete beta function is \( B(a+1, b+1) = B_1(a+1, b+1) \). Omitting coefficients, \( H_m(w) \) can be expressed as:
\[
\begin{aligned}
    & \textstyle \left( \frac{1}{m - 1} - \frac{1}{2m - 1} \right ) \left [ B(n - m, m + 1) - B_w(n - m, m + 1) \right ] \\
    & \textstyle + \left( \frac{1}{m - 1} - \frac{1}{2m - 1} \right ) B_w(n - m, m + 1) \textstyle - \left( \frac{1 - w^{m - 1}}{m - 1} - \frac{1 - w^{2m - 1}}{2m - 1} \right ) B_w(n - m, m + 1) \\
    & \textstyle + \int_0^w \left [ B_w(n - m, m + 1) - B_{\frac{w - z}{1 - z}}(n - m, m + 1) \right ] \cdot[1 - z^m] z^{m-2} \, dz.
\end{aligned}
\]

After canceling identical terms, the expression simplifies to:
\[
\begin{aligned}
    \left( \frac{1}{m - 1} - \frac{1}{2m - 1} \right ) B(n - m, m + 1) & - \int_0^w B_{\frac{w - z}{1 - z}}(n - m, m + 1) [1 - z^m] z^{m-2} \, dz.
\end{aligned}
\]

The previously omitted coefficient is:
\[
\begin{aligned}
    n (m - 1) \binom{n - 1}{m} & = \frac{n! (m - 1)}{(n - m - 1)! m!} \\
    & = (m - 1) B^{-1}(n - m, m + 1).
\end{aligned}
\]

Substituting this back, we obtain the expression for \( H_m(w) \):
\[
\frac{m}{2m - 1} - (m - 1) \int_0^w I_{\frac{w - z}{1 - z}}(n - m, m + 1) [1 - z^m] z^{m-2} \, dz,
\]
where \( I_w(a+1, b+1) = \frac{B_w(a+1, b+1)}{B(a+1, b+1)} \) is the regularized incomplete beta function, representing the cumulative distribution function of the beta distribution.
\end{proof}

\subsection{Proof of Lemma~7}
\begin{proof}
    The cumulative distribution function of the Beta distribution is expressed as:
    \[
    \begin{aligned}
        I_x(n-m,m+1) & = \int_0^x \beta_t(n-m,m+1) \, dt,
    \end{aligned}
    \]
    where \( \beta_t(n-m,m+1) = \frac{t^{n-m-1}(1-t)^{m}}{B(n-m,m+1)} \).

    We first characterize the unimodal property of the function \( \beta_t(n-m,m+1) \). Denoting it simply as \(\beta(t)\), we have:
    \[ 
    \beta'(t) = t^{n-m-2} (1-t)^{m-1} \left[ (n - m - 1)(1 - t) - m t \right], 
    \]
    which implies three critical points. Let \( t_0 = \frac{n-m-1}{n-1} \), then the critical points are \(\beta(0) = 0\), \(\beta(1) = 0\), and \(\beta(t_0) > 0 \). Since the third term in the derivative is a linear function of \( t \), it can be easily verified that \( \beta'(t) > 0 \) for \( t \in ( 0, t_0 ) \) and \( \beta'(t) < 0 \) for \( t \in (t_0, 1) \). Therefore, the function \(\beta\) attains its unique maximum at \( t = t_0 = 1- \frac{m}{n-1} \).

    Next, we characterize the concentration property of the function \( \beta(t) \) around its peak.

    Perform a second-order Taylor expansion approximation of \( \ln \beta(t) \) at \( t_0 \):
    \[
    \begin{aligned}
    & \ln \beta(t) =  \ln \beta(t_0) + \frac{d}{dt} \ln \beta(t) \bigg |_{t_0}(t-t_0) +  \frac{1}{2} \frac{d^2}{dt^2} \ln \beta(t) \bigg |_{t_0} (t-t_0)^2.
\end{aligned}
    \]

    The second derivative of \( \beta(t) \) at \( t_0 \) is:
    \[
    \begin{aligned}
        \frac{d^2}{dt^2} \ln \beta(t) \bigg |_{t_0} & = - \left [ \frac{n-m-1}{t_0^2} + \frac{m}{(1-t_0)^2} \right ] \\
        & = - \frac{(n-1)^3}{(n-m-1)m},
    \end{aligned}
    \]
    and \( \frac{d}{dt} \ln \beta(t) \bigg |_{t_0} = 0 \). Thus, in the neighborhood of \( t_0 \), we have:
    \[
    \begin{aligned}
        \beta(t) & = \beta(t_0) \cdot \exp\left \{ -\frac{(n-1)^3}{(n-m-1)m}(t-t_0)^2 \right \} \\
        & \leq \beta(t_0) \cdot \exp\left \{ -4(n-1)(t-t_0)^2 \right \},
    \end{aligned}
    \]
    where the last inequality follows from the arithmetic-geometric mean inequality \( (n-m-1)m \leq (n-1)^2/4 \).

    Now we obtain that when \( |t-t_0| \geq \delta > n^{-\frac{1}{4}} \), i.e., there exists \( N_1 \) such that for \( n > N_1 \), \( n^{-1/4} < \delta \), we have \( \beta(t) \leq \beta(t_0) \cdot \exp \left \{ -\sqrt{n}\right \} \).

    Moreover:
    \[
    \begin{aligned}
    \beta(t_0) = & (1-\frac{m}{n-1})^{n-m-1} \cdot (\frac{m}{n-1})^m \cdot B(n-m,m+1)^{-1} \\
    = & \frac{(n-m-1)^{n-m-1}m^m}{(n-1)^{n-1} m!} \cdot \frac{n!}{(n-m-1)!} \\
    \leq & \frac{m^m}{m!} \cdot \frac{(n-m-1)^{n-m-1}}{(n-1)^{n-1}} \cdot \frac{\sqrt{2\pi n} n^n e^{-(m+1)}}{\sqrt{2\pi (n-m-1)}(n-m-1)^{n-m-1}} \\
    = & \frac{m^m}{m!} \cdot (1+\frac{1}{n-1})^{n-1} \cdot n^{3/2} \cdot \frac{e^{-(m+1)}}{\sqrt{n-m-1}} \\
    \leq & \frac{m^m}{\frac{1}{2} \sqrt{2\pi m}(\frac{m}{e})^{m}} \cdot (1+\frac{1}{n-1})^{n-1} \cdot n^{3/2} \cdot \frac{e^{-(m+1)}}{\sqrt{n-m-1}} \\
    = & \frac{2n^{3/2}}{\sqrt{2\pi m(n-m-1)}} = O(n^{\frac{3}{2}}),
\end{aligned}
    \]
    where the two inequalities use the original form of Stirling's approximation and its half-form to ensure the bounds hold.

    Therefore, we obtain \( \lim_{n \to \infty} \beta(t_0)\cdot \exp\{ -\sqrt{n} \} = 0 \), i.e., there exists \( N_2 \) such that for \( n > N_2 \), \(\beta(t_0)\cdot \exp\{ -\sqrt{n}\} < \epsilon \) holds.

    Let \( N = \max \{ N_1, N_2 \} \). Then when \( x < t_0 - \delta \), \( |x-t_0| \geq \delta \) holds. Thus, for \( n > N \), any \( \epsilon >0 \), \( \delta >0 \), and \( m \geq 2 \), we have:
    \[
    \begin{aligned}
        I_x(n-m, m+1) & = \int_0^x \beta(t) \, dt \\
        & \leq \int_0^1 \beta(t_0) \, dt \leq \epsilon.
    \end{aligned}
    \]

    Meanwhile, since \( I_x \) is a cumulative distribution function, \( I_x(n-m,m+1) \leq 1 \) always holds trivially for \( x \in [0,1] \). 

    In conclusion, for any \( \epsilon > 0 \), \( \delta > 0 \), we have found an \( N \in \mathbb{N} \) such that when \( n > N \), for any arbitrary \( m \geq 2 \), \( I_x(n-m,m+1) \leq \epsilon \) holds for \( x \in [0, 1-\frac{m}{n-1}-\delta) \) and \( I_x(n-m,m+1) \leq 1 \) holds for \( x \in [1-\frac{m}{n-1}-\delta, 1] \). This completes the proof.
\end{proof}

\subsection{Proof of Theorem~2}
\begin{proof}
Given that the optimal reward structure is winner-take-all, we need to determine the optimal number of qualifiers \( m \) that maximizes the expected Highest Individual Performance:
    \[
        \mathbb{E}_{y \sim X^{(n-m)}} \big [ \mathbb{E}_{x^{(n)} \sim X^{(n)} \mid y} [ b(x^{(n)} ; y) ] \big ].
    \]

According to Proposition~\ref{prop:HighestWTA}, this expression expands to:
\[
\begin{aligned}
    \int_0^\infty \left ( \int_y^\infty x (m-1) [1-P(x)^m] P(x)^{m-2} p(x) \, dx \right ) \\ 
    \cdot n \binom{n-1}{m} F(y)^{n-m-1}(1-F(y))^m f(y) \, dy.
\end{aligned}
\]

To simplify this expression, we first perform a change of variables, then switch the integration order. The calculation is detailed in Lemma~\ref{lem:iso2}, and we eventually have:
\begin{align*}
    \textstyle\int_0^1 x'(w) \left ( \int_w^1 \int_0^v n (m-1) \binom{n-1}{m} \left[1-\left(\frac{v-u}{1-u}\right)^m\right] \right. \left. \cdot(v-u)^{m-2} u^{n-m-1} (1-u) \, du \, dv \right ) \, dw.
\end{align*}


Denote the second term in the outer integral as \( H_m(w) = \int_w^1 G_m(v) \, dv\). The equality holds because:
\[
\begin{aligned}
    \int_0^1 G_m(v) x(v) \,dv & = \int_0^1 G_m(v) \left ( \int_0^v x'(w) \, dw\right ) \, dv \\
    & = \int_0^1 \int_w^1 G_m(v) x'(w) \, dv \, dw \\
    & = \int_0^1 x'(w)\left ( \int_w^1 G_m(v) \, dv\right ) \, dw.
\end{aligned}
\]

We now simplify the expression for \( H_m(w) \). After a series of transformation in Lemma~\ref{lem:betarep}, we finally obtain:
\[
\frac{m}{2m - 1} - (m - 1) \int_0^w I_{\frac{w - z}{1 - z}}(n - m, m + 1) [1 - z^m] z^{m-2} \, dz,
\]
where \( I_w(a+1, b+1) = \frac{B_w(a+1, b+1)}{B(a+1, b+1)} \) is the regularized incomplete beta function, representing the cumulative distribution function of the beta distribution.

We now analyze the asymptotic behavior of this expression, which relies on the concentration properties of the beta distribution.
Using Lemma~\ref{lem:AsmStep}, we can bound \( I_{\frac{w - z}{1 - z}}(n - m, m + 1) \) with a step function, thereby characterizing an upper bound for the integral containing it. 

Let \( \bar{x'} \) be an upper bound for \( x'(w) \) on \([0,1]\). We now analyze the integral by cases:

\underline{Case 1:} For \( w \in [0, 1 - k - \delta) \), we have \( \frac{w - z}{1 - z} < \frac{1 - k - \delta - z}{1 - z} = 1 - \frac{k + \delta}{1 - z} \leq 1 - k - \delta \). Thus, \( I_{\frac{w - z}{1 - z}} \leq \epsilon \) over this integration region. The integral becomes:
\[
\begin{aligned}
    \int_{0}^{1 - k - \delta} x'(w) \int_0^{w} I_{\frac{w - z}{1 - z}}(n - m, m + 1) & \quad \cdot (m - 1) [1 - z^m] z^{m-2} \, dz \, dw \\
    \leq & \bar{x'} \epsilon \int_{0}^{1 - k - \delta} \int_0^{w} (m - 1) [1 - z^m] z^{m-2} \, dz \, dw \\
    \leq & \bar{x'} \epsilon \int_{0}^{1 - k - \delta} \int_0^{w} (m - 1) z^{m-2} \, dz \, dw \\
    \leq & \bar{x'} \epsilon \int_{0}^{1} \int_0^{1} (m - 1) z^{m-2} \, dz \, dw = \bar{x'} \epsilon.
\end{aligned}
\]

\underline{Case 2:} For \( w \in [1 - k - \delta, 1] \), if \( \frac{w - z}{1 - z} \geq 1 - k - \delta \), this requires \( z \leq 1 - \frac{1 - w}{k + \delta} \). The integral in this case is:
\[
\begin{aligned}
    \int_{1 - k - \delta}^{1} x'(w) \int_0^{1 - \frac{1 - w}{k + \delta}} I_{\frac{w - z}{1 - z}}(n - m, m + 1) & \quad \cdot(m - 1) [1 - z^m] z^{m-2} \, dz \, dw \\
    \leq & \bar{x'} \int_{1 - k - \delta}^{1} \bigg ( \int_0^{1 - \frac{1 - w}{k + \delta}} (m - 1) z^{m-2} \, dz \bigg ) \, dw \\
    = & \bar{x'} \int_{1 - k - \delta}^{1} \left( 1 - \frac{1 - w}{k + \delta} \right)^{m-1} \, dw.
\end{aligned}
\]
Let \( u = 1 - \frac{1 - w}{k + \delta} \), then \( dw = (k + \delta) du \), with \( u = 0 \) when \( w = 1 - k - \delta \) and \( u = 1 \) when \( w = 1 \). Substituting yields:
\[
\begin{aligned}
    = & \bar{x'} \int_{0}^1 (k + \delta) u^{m-1} \, du \\
    = & \bar{x'} \frac{k + \delta}{m} \leq \bar{x'} \left( \frac{1}{n - 1} + \delta \right).
\end{aligned}
\]

For \( z > 1 - \frac{1 - w}{k + \delta} \), where \( I_{\frac{w - z}{1 - z}} \leq \epsilon \), the integral becomes:
\[
\begin{aligned}
    \int_{1 - k - \delta}^{1} x'(w) \int_{1 - \frac{1 - w}{k + \delta}}^{w} I_{\frac{w - z}{1 - z}}(n - m, m + 1) & \quad \cdot (m - 1) [1 - z^m] z^{m-2} \, dz \, dw \\
    \leq & \bar{x'} \epsilon \int_{1 - k - \delta}^1 \bigg ( \int_{1 - \frac{1 - w}{k + \delta}}^w (m - 1) z^{m-2} \, dz \bigg )\, dw \\
    \leq & \bar{x'} \epsilon \int_{0}^1 \bigg ( \int_{0}^1 (m - 1) z^{m-2} \, dz \bigg )\, dw \\
    = & \bar{x'} \epsilon.
\end{aligned}
\]

Combining these results, we conclude that for sufficiently large \( n \):
\[
\begin{aligned}
    (m - 1) \int_0^w I_{\frac{w - z}{1 - z}}(n - m, m + 1) [1 - z^m] z^{m-2} \, dz \leq & \bar{x'} \left( \frac{1}{n - 1} + 2 \epsilon + \delta \right ) \to 0.
\end{aligned}
\]

Thus, for any \( m \geq 2 \), we have:
\[
\lim_{n \to \infty} S^{(n)}(n, m, 1) = \frac{m}{2m - 1} \int_0^1 x'(w) \, dw,
\]
with the convergence rate independent of \( m \).

This allows us to directly compare the highest performance for different shortlist sizes \( m \). Since this expression decreases monotonically with \( m \), the optimum occurs at \( m = 2 \), i.e., a two-contestant winner-take-all contest. The highest performance in this case is:
\[
\lim_{n \to \infty} S^{(n)}(n, 2, 1) = \frac{2}{3} \bar{x},
\]
which completes the proof.
\end{proof}
\subsection{Proof of Theorem~3}
\begin{proof}
    Theorem~\ref{thm:AsmHP} gives the asymptotic best contest with pre-selection for highest individual performance objective, i.e., a 2-contestant winner-take-all contest. 
    
    Also, previous work has shown that, without pre-selection, the optimal contest is an \( n \)-contestant winner-take-all format, yielding an asymptotic highest performance of \( \frac{1}{2} \bar{x} \) for the designer \cite{AS09}. 

    They directly give:
    \[
    \begin{aligned}
        \lim_{n \rightarrow \infty} \frac{S^{(n)}(m^*,n,1)}{S^{(n)}(n,n,1)} & = \frac{\lim_{n \rightarrow \infty}S^{(n)}(m^*,n,1)}{\lim_{n \rightarrow \infty}S^{(n)}(n,n,1)} \\
        & = \frac{\lim_{n \rightarrow \infty}S^{(n)}(2,n,1)}{\frac{1}{2}\bar{x}} \\
        & = \frac{\frac{2}{3} \bar{x}}{\frac{1}{2}\bar{x}} = \frac{4}{3},
    \end{aligned}
    \]
    as desired, which completes the proof.
\end{proof}

\subsection{Proof of Proposition~4}
\begin{proof}
    Similar to Theorem~\ref{thm:AsmHP}, we have
    \[
    \begin{aligned}
    S^{(n)}(m,n,1) = & \int_0^1 x'(w) \bigg ( \int_w^1 \int_0^v n (m-1) \binom{n-1}{m} 
     [1-(\frac{v-u}{1-u})^m] \\
    & \cdot(v-u)^{m-2} u^{n-m-1} (1-u) \, du \, dv \bigg ) \, dw,
\end{aligned}
    \]
    where $x'(w)=b$ for the $U[0,b]$ distribution.
    We prove that $S^{(n)}(m,n,1)$ is decreasing in $m$.
    
    By some calculation, we have
\begin{align*}
    & \frac1{b}S(m,n,1) \\
    =&\int_0^1  \int_w^1 \int_0^v n (m-1) \binom{n-1}{m} [1-(\frac{v-u}{1-u})^m] \cdot (v-u)^{m-2} u^{n-m-1} (1-u) \, du \, dv \, dw \\
    =&\int_{u=0}^1\int_{v=u}^1 \int_{w=0}^vn (m-1) \binom{n-1}{m} [1-(\frac{v-u}{1-u})^m]  \cdot (v-u)^{m-2} u^{n-m-1} (1-u) dw\, dv \, du \\
    =&\int_{u=0}^1\int_{v=u}^1 n (m-1) \binom{n-1}{m} [1-(\frac{v-u}{1-u})^m] \cdot (v-u)^{m-2} u^{n-m-1} (1-u) v\, dv \, du \\
    =&\int_{u=0}^1\int_{\theta=0}^1 n (m-1) \binom{n-1}{m} (1-\theta^m) ((1-u)\theta)^{m-2} \\
    & \cdot u^{n-m-1} (1-u) (u+(1-u)\theta)(1-u)d\theta du \\
    =&\int_{u=0}^1\int_{\theta=0}^1 n (m-1) \binom{n-1}{m} ((1-\theta^m) \theta^{m-2} u^{n-m} \\
    & \cdot (1-u)^m+(1-\theta^m) \theta^{m-1}u^{n-m-1}(1-u)^{m+1})d\theta du,
\end{align*}
where the fourth equality is by substituting $\theta=\frac{v-u}{1-u}$, i.e., $v=u+(1-u)\theta$.

Since $\int_0^1(1-\theta^m)\theta^{m-2}d\theta=\frac1{m-1}-\frac1{2m-1}$, $\int_0^1(1-\theta^m)\theta^{m-1}d\theta=\frac1{m}-\frac1{2m}$, we have
\begin{align*}
    \frac1{b}S(m,n,1) = &\int_{0}^1 n (m-1) \binom{n-1}{m} ((\frac1{m-1}-\frac1{2m-1})u^{n-m} \\
    & \cdot(1-u)^m+(\frac1{m}-\frac1{2m})u^{n-m-1}(1-u)^{m+1})du \\
    = & \int_{0}^1 n \binom{n-1}{m} (\frac{m}{2m-1}u^{n-m}(1-u)^m +\frac{m-1}{2m}u^{n-m-1}(1-u)^{m+1})du \\
    = & \int_{0}^1\frac{n!}{(n-m-1)!m!}(\frac{m}{2m-1}u^{n-m}(1-u)^m + \frac{m-1}{2m}u^{n-m-1}(1-u)^{m+1})du.
\end{align*}

By Lemma \ref{lem:betaInt} we have $\int_0^1u^{n-m}(1-u)^mdu=\frac{(n-m)!m!}{(n+1)!}$ and $\int_0^1u^{n-m-1}(1-u)^{m+1}du=\frac{(n-m-1)!(m+1)!}{(n+1)!}$, and it follows that
\begin{align*}
    & \frac1{b}S(m,n,1) = \frac{n!}{(n-m-1)!m!}(\frac{m}{2m-1}\cdot\frac{(n-m)!m!}{(n+1)!} +\frac{m-1}{2m}\cdot\frac{(n-m-1)!(m+1)!}{(n+1)!}) \\
    & =\frac{m}{2m-1}\cdot\frac{n-m}{n+1}+\frac{m-1}{2m}\cdot\frac{m+1}{n+1} \\
    & =\frac{2m^2(n-m)+(m^2-1)(2m-1)}{2m(2m-1)(n+1)} \\
    & =\frac{2m^2n-m^2-2m+1}{2m(2m-1)(n+1)} \\
    & =\frac{2m^2(n-\frac52)}{2m(2m-1)(n+1)}+\frac{4m^2-2m+1}{2m(2m-1)(n+1)} \\
    & =\frac{m}{2m-1}\cdot\frac{n-\frac52}{n+1}+\frac1{n+1}+\frac1{2m(2m-1)(n+1)}.
\end{align*}
When $n=2$, $m=2$ is the only valid choice. When $n\geq 3$, we have $n-\frac52>0$, and it follows that both $\frac{m}{2m-1}\cdot\frac{n-\frac52}{n+1}$ and $\frac1{2m(2m-1)(n+1)}$ are strictly decreasing in $m$. That is, $S(m,n,1)$ is decreasing in $m$.
\end{proof}

\subsection{Proof of Proposition~5}
We first give two technical lemmas to aim our proof.
\begin{applemma}\label{lem:EqSum}
Let \(X_1, X_2, \ldots, X_n\) be independent and identically distributed continuous random variables with cumulative distribution function \(F\) and continuous probability density function \(f\). Denote the order statistics as \(X^{(1)} < X^{(2)} < \cdots < X^{(n)}\), where \(X^{(k)}\) represents the \(k\)-th order statistic. Then the identity \(\sum_{i=1}^{m} f(X^{(n-i+1)} = x \mid X^{(n-m)} = x^{(n-m)}) = m f(X = x \mid X > x^{(n-m)})\) holds universally.
\end{applemma}
\begin{proof}
Let \(a = x^{(n-m)}\). Given the condition \(X^{(n-m)} = a\), the order statistics \(X^{(n-m+1)}, X^{(n-m+2)}, \ldots, X^{(n)}\) correspond to the order statistics of the \(m\) observations greater than \(a\). Let \(Y_1, Y_2, \ldots, Y_m\) be independent and identically distributed continuous random variables with distribution \(F\) truncated on \((a, \infty)\), i.e., with density function:
\[
f_Y(y) = \frac{f(y)}{1 - F(a)}, \quad y > a.
\]
Denote their order statistics by \(Y^{(1)} < Y^{(2)} < \cdots < Y^{(m)}\). 

By Lemma~\ref{lem:Truncate}, the conditional distribution of \(\left(X^{(n-m+1)}, X^{(n-m+2)}, \ldots, X^{(n)}\right) \mid X^{(n-m)} = a\) is identical to that of \(\left(Y^{(1)}, Y^{(2)}, \ldots, Y^{(m)}\right)\). 

Now, consider the summation on the left-hand side:
\[
\sum_{i=1}^{m} f\left(X^{(n-i+1)} = x \mid X^{(n-m)} = a\right).
\]
When \(i = 1\), \(n - i + 1 = n\), i.e., \(X^{(n)}\) is the maximum; when \(i = m\), \(n - i + 1 = n - m + 1\), i.e., \(X^{(n-m+1)}\) is the minimum among the observations greater than \(a\). Thus, the summation corresponds to \(X^{(n)}, X^{(n-1)}, \ldots, X^{(n-m+1)}\), which are equivalent to \(Y^{(m)}, Y^{(m-1)}, \ldots, Y^{(1)}\).

Equivalently, the summation becomes:
\[
\sum_{k=1}^{m} f_{Y^{(k)}}(x),
\]
where \(f_{Y^{(k)}}(x)\) is the marginal density function of \(Y^{(k)}\).

For the \(k\)-th order statistic \(Y^{(k)}\) of \(m\) i.i.d. continuous random variables, its density function is:
\[
\begin{aligned}
    f_{Y^{(k)}}(y) = &  \frac{m!}{(k-1)! (m-k)!} \left[ \frac{F(y) - F(a)}{1 - F(a)} \right]^{k-1} \cdot\left[ \frac{1 - F(y)}{1 - F(a)} \right]^{m-k} \frac{f(y)}{1 - F(a)}, \quad y > a.
\end{aligned}
\]
Let:
\[
p = \frac{F(x) - F(a)}{1 - F(a)}, \quad q = \frac{1 - F(x)}{1 - F(a)}, \quad \text{so that} \quad p + q = 1.
\]
Substituting these, we obtain:
\[
f_{Y^{(k)}}(x) = \frac{m!}{(k-1)! (m-k)!} p^{k-1} q^{m-k} \frac{f(x)}{1 - F(a)}.
\]
Note that:
\[
\frac{m!}{(k-1)! (m-k)!} = m \binom{m-1}{k-1},
\]
since
\(
\binom{m-1}{k-1} = \frac{(m-1)!}{(k-1)! (m-k)!}\), and thus \(m \binom{m-1}{k-1} = m \frac{(m-1)!}{(k-1)! (m-k)!} = \frac{m!}{(k-1)! (m-k)!}\).
Therefore:
\[
f_{Y^{(k)}}(x) = m \binom{m-1}{k-1} p^{k-1} q^{m-k} \frac{f(x)}{1 - F(a)}.
\]
Now, summing over \(k\):
\[
\begin{aligned}
    \sum_{k=1}^{m} f_{Y^{(k)}}(x) & = \sum_{k=1}^{m} m \binom{m-1}{k-1} p^{k-1} q^{m-k} \frac{f(x)}{1 - F(a)} \\ & = m \frac{f(x)}{1 - F(a)} \sum_{k=1}^{m} \binom{m-1}{k-1} p^{k-1} q^{m-k}.
\end{aligned}
\]
Let \(j = k - 1\). Then when \(k = 1\), \(j = 0\), and when \(k = m\), \(j = m-1\). The summation becomes:
\[
\sum_{j=0}^{m-1} \binom{m-1}{j} p^j q^{(m-1)-j} = (p + q)^{m-1} = 1^{m-1} = 1,
\]
since \(p + q = 1\). Thus:
\[
\sum_{k=1}^{m} f_{Y^{(k)}}(x) = m \frac{f(x)}{1 - F(a)} \cdot 1 = m \frac{f(x)}{1 - F(a)}.
\]
The right-hand side is:
\[
m f\left(X = x \mid X > x^{(n-m)}\right) = m f\left(X = x \mid X > a\right).
\]
Given \(X > a\), the conditional density of a single random variable \(X\) is:
\[
f\left(X = x \mid X > a\right) = \frac{f(x)}{1 - F(a)}, \quad x > a.
\]
Hence:
\[
m f\left(X = x \mid X > a\right) = m \frac{f(x)}{1 - F(a)}.
\]
Therefore:
\[
\begin{aligned}
    & \sum_{i=1}^{m} f\left(X^{(n-i+1)} = x \mid X^{(n-m)} = a\right) \\ = & m f\left(X = x \mid X > a\right).
\end{aligned}
\]
When \(x \leq a\), both sides equal zero (since the order statistics \(X^{(n-i+1)} > a\) and the condition \(X > a\) requires \(x > a\)), so the equality still holds. This completes the proof.
\end{proof}

\begin{applemma} \label{lem:EqExpSum}
    Let \(X_1, X_2, \ldots, X_n\) be independent and identically distributed continuous random variables with cumulative distribution function \(F\) and continuous probability density function \(f\). Let \( h(\cdot;X^{(n-m)}) \) be a function parameterized by the order statistic \( X^{(n-m)} \). Under the condition \( X^{(n-m)} = a \) and letting \( r = n-m \), the following equality holds:
    \[
    \begin{aligned}
        & \mathbb{E}_{x^{(r+1)}, \ldots,x^{(n)} \sim X^{(r+1)}, \ldots X^{(n)} \mid x^{(r)}} \left [ \sum_{l=1}^{m} h(x^{(r+l)} ; x^{(r)})  \right ] \\
        = & m\mathbb{E}_{x \sim X \mid X >x^{(r)}} \left [ h(x;x^{(r)})\right ].
    \end{aligned}
    \]
\end{applemma}
\begin{proof}
    By the linearity of expectation, the original expression can be transformed into:
    \[
    \sum_{l=1}^{m}\mathbb{E}_{x^{(r+1)}, \ldots,x^{(n)} \sim X^{(r+1)}, \ldots X^{(n)} \mid x^{(r)}} \left [ h(x^{(r+l)} ; x^{(r)})  \right ].
    \]
    Expanding the expectation by definition and marginalizing yields:
    \[
    \begin{aligned}
        & \sum_{l=1}^{m} \int_{x^{(r+1)}<\ldots<x^{(n)}} h(x^{(r+l)} ; x^{(r)}) f(\ldots) \, dx_{r+1} \ldots dx_{n} \\
        = & \sum_{l=1}^{m} \int_{x^{(r)}}^{\infty} h(x ; x^{(r)}) f(X^{(r+l)}=x| X^{(r)} = x^{(r)}) \, dx \\
        = & \int_{x^{(r)}}^{\infty} h(x ; x^{(r)}) \left( \sum_{l=1}^{m}f(X^{(r+l)}=x| X^{(r)} = x^{(r)}) \right ) \, dx.
    \end{aligned}
    \]
    From Lemma~\ref{lem:EqSum}, we have \( \sum_{l=1}^{m} f(X^{(r+l)} = x | X^{(r)} = x^{(r)}) = m f( X = x | X > x^{(r)}) \), thus:
    \[
    \begin{aligned}
        = & \int_{x^{(r)}}^{\infty} h(x ; x^{(r)}) m f( X = x | X > x^{(r)}) \, dx \\ 
        = & m\mathbb{E}_{x \sim X \mid X >x^{(r)}} \left [ h(x;x^{(r)})\right ].
    \end{aligned}
    \]
    This completes the proof.
\end{proof}

    Now we are ready to prove the Proposition~5.
\begin{proof}
    Given the number of advancing contestants $m$ and the elimination cutoff $y$, the total performance can be expressed as:
    \[
    \begin{aligned}
    & \mathbb{E}_{x^{(n-m+1)}, \ldots,x^{(n)} \sim X^{(n-m+1)}, \ldots X^{(n)} \mid x^{(n-m)}} \left [ \sum_{l=1}^{m} b(x^{(n-l+1)} ; x^{(n-m)}) \right ].
\end{aligned}
    \]

    By Lemma~\ref{lem:EqExpSum}, this expression is equivalent to:
    \[
    m\mathbb{E}_{x \sim X \mid X >x^{(n-m)}} \left [ b(x;x^{(n-m)})\right ],
    \]
    which can be interpreted as the sum of performances of $m$ contestants independently drawn from the truncated distribution $P$, where each contestant's expected performance is:
    \[
    \begin{aligned}
        & \mathbb{E}_{X\sim P}[b(x \mid y)] \\
        = &  \int_0^\infty \bigg ( \int_0^x \sum_{l=1}^{m-1}Z_l\binom{m-1}{l-1}\frac{m-l}{l}(1-P(t))^{l-1} \cdot P(t)^{m-l-1}p(t)t\, dt \bigg ) p(x) \, dx\\
        = &  \int_0^\infty \bigg ( \int_0^x \sum_{l=1}^{m-1}Z_l\binom{m-1}{l}(1-P(t))^{l-1} \cdot P(t)^{m-l-1}p(t)t\, dt \bigg ) p(x) \, dx\\
        = & \sum_{l=1}^{m-1}Z_l \int_0^\infty \bigg ( \int_0^x \binom{m-1}{l}(1-P(t))^{l-1}
         \cdot P(t)^{m-l-1}p(t)t\, dt \bigg ) p(x) \, dx. 
    \end{aligned}
    \]

    By changing the order of integration, we have:
    \[
    \begin{aligned}
    = & \sum_{l=1}^{m-1}Z_l \int_0^\infty \int_t^\infty \binom{m-1}{l}(1-P(t))^{l-1} \cdot P(t)^{m-l-1}p(t)t p(x) \, dx \, dt \\
    = & \sum_{l=1}^{m-1}Z_l \int_0^\infty \binom{m-1}{l}(1-P(t))^{l-1} \cdot P(t)^{m-l-1}p(t)t \left ( \int_t^\infty p(x) \, dx \right ) dt \\
    = & \sum_{l=1}^{m-1}Z_l \int_0^\infty \binom{m-1}{l}(1-P(t))^{l} \cdot P(t)^{m-l-1}p(t)t \, dt.
\end{aligned}
    \]
    
    Using the probability density function of order statistics $f_{k,n}(x) = n \binom{n-1}{k-1}F(x)^{k-1}(1-F(x))^{n-k} f(x)$, the integral can be rewritten as:
    \[
    \begin{aligned}
    & = \frac{1}{m} \sum_{l=1}^{m-1}Z_l \int_0^\infty t f_{m-l,m}(t) dt 
    \\ & = \frac{1}{m}\sum_{l=1}^{m-1}Z_l\mathbb{E}_{x\sim P^{(m-l)}}[x].
    \end{aligned}
    \]

    We aim to maximize this expression under the constraint $\sum_{l=1}^{m-1} Z_l = 1$. By the first-order stochastic dominance property of order statistics, we have $\mathbb{E}_{x\sim P^{(m-l)}}[x] > \mathbb{E}_{x\sim P^{(m-l-1)}}[x]$. Therefore, the expression is maximized when $Z_l = 1$.

    Since $Z_l = 1$ yields the maximum total performance for any realization of $y$, it remains optimal after taking expectation over the elimination cutoff $y \sim X^{(n-m)}$.

    In conclusion, when $m$ is fixed, the optimal prize structure must be $V_1=1$, i.e., a winner-take-all contest.
\end{proof}

\subsection{Proof of Proposition~6}

We first introduce a technical lemma:
\begin{applemma}\label{lem:SimpleSame}
For any number of registered contestants \( n \) and ability distribution \( F \), when the prize structure is a simple contest with \( l \) prizes (\( l < m \)), i.e., \( V_1 = \ldots = V_l > V_{l+1} = 0 \), the total performance remains constant regardless of the shortlist size \( m \). This total performance is given by \( Z_l\mathbb{E}_{X^{(n-l)}}[X] \), where \( Z_l = lV_l \).
\end{applemma}
\begin{proof}
    For any simple contest with shortlist size \( m \) and cut-off score \( y \), the total performance is $\mathbb{E}_{X \sim P^{(m-l)}}[X]$. Therefore, the ex-ante total performance expands as:
\[
\begin{aligned}
& \textstyle \int_0^\infty\left(\int_y^\infty x \, dP^{(m-l)}(x)\right) \, dF^{(n-m)} (y)\\
= & \textstyle \int_0^\infty\left(\int_y^\infty x m \binom{m-1}{m-l-1} P(x)^{m-l-1}(1-P(x))^{l}p(x) \, dx\right) \, \\
& \textstyle \cdot n \binom{n-1}{n-m-1} F(y)^{n-m-1}(1-F(y))^m f(y) \, dy \\
= & \textstyle \int_0^\infty \bigg ( \int_y^\infty x m \binom{m-1}{m-l-1} ( \frac{F(x)-F(y)}{1-F(y)} )^{m-l-1} (\frac{1-F(x)}{1-F(y)})^{l} \\
& \textstyle \cdot \frac{f(x)}{1-F(y)} \, dx\bigg ) \, n \binom{n-1}{n-m-1} F(y)^{n-m-1}(1-F(y))^m f(y) \, dy \\
\end{aligned}
\]
\[
\begin{aligned}
    = & \textstyle n \binom{n-1}{n-m-1}m \binom{m-1}{m-l-1} \int_0^\infty \bigg ( \int_0^x  [F(x)-F(y)]^{m-l-1} \\
    & \cdot F(y)^{n-m-1} f(y) \, dy \bigg ) \, x [1-F(x)]^l f(x) \, dx \\
    = & \textstyle n \binom{n-1}{n-m-1}m \binom{m-1}{m-l-1} \int_0^\infty \bigg( \int_0^{F(x)}  [F(x)-t]^{m-l-1} \cdot t^{n-m-1} \, dt \bigg) x [1-F(x)]^l f(x) \, dx,
\end{aligned}
\]
where the last equality uses the substitution \( t = F(y) \).

By Lemma~\ref{lem:betaInt}, the inner integral evaluates to $\frac{(m-l-1)!(n-m-1)!}{(n-l-1)!}F(x)^{n-l-1}$. The total coefficient becomes:
\[
\begin{aligned}
    \textstyle n \frac{(n-1)!}{(n-m-1)!m!} \frac{(m-l-1)!(n-m-1)!}{(n-l-1)!} m \frac{(m-1)!}{(m-l-1)!l!} \textstyle = & \textstyle n \frac{(n-1)!}{(n-l-1)!l!} = n \binom{n-1}{n-l-1}
\end{aligned}
\]
Substituting back into the original expression yields:
$$
\begin{aligned}
    \textstyle \int_0^\infty x \left( n \binom{n-1}{n-l-1} F(x)^{n-l-1} (1-F(x))^{l} f(x)\right) \, dx = & \textstyle \int_0^\infty x \, dF^{(n-l)} = \mathbb{E}_{X^{(n-l)}}[X].
\end{aligned}
$$

Thus, regardless of the shortlist size, the ex-ante total performance equals the expected ability value of the \((l+1)\)-th ranked participant.
\end{proof}

    We now ready to prove Proposition~6.
\begin{proof}
    Let \( Z_l = l(V_l-V_{l+1}) \). Under the condition that the cost function \( g \) is linear and the objective function is linear in performance, the effects of different adjacent award pairs \( (V_l, V_{l+1}) \) on performance are additive. Therefore, \( Z_l \) is also additive, and we can consider the impact of each \( Z_l \) separately on performance. 
    
    A contest with only \( Z_l > 0 \) is equivalent to a simple contest where \( V'_1 = \ldots = V'_l = \frac{Z_l}{l} > V'_{l+1} = 0 \). According to Lemma~\ref{lem:SimpleSame}, the corresponding total performance is \( Z_l\mathbb{E}_{X^{(n-l)}}[X] \). Thus, the total performance for contest \( (m, \vec{V}) \) is \( \sum_{l=1}^{m-1} Z_l\mathbb{E}_{X^{(n-l)}}[X] \), which is independent of the shortlist size.
\end{proof}

\subsection{Proof of Theorem~4}
\begin{proof}
    It follows directly from Propositions~\ref{prop:TotalWTA} and Proposition~\ref{prop:Allsame} that when pre-selection is applied (i.e., \(m < n\)), any winner-take-all contest achieves the same total performance, \( \mathbb{E}_{X^{(n-1)}}[X] \). Furthermore, the optimal contest without pre-selection, the \( n \)-contestant winner-take-all contest—also attains total performance \( \mathbb{E}_{X^{(n-1)}}[X] \) \cite{MS01}, completing the proof.
\end{proof}

\section{A Non-Optimality Example for (\(m=2\))}
\label{apx:non-optimal}

The numerical evidence in the main text suggests that pre-selection combined with a two-contestant winner-take-all contest is highly robust across a wide range of common ability distributions. However, this property is not universal. We construct here a counterexample showing that in finite populations the optimal shortlist size may exceed two under extreme heterogeneity.

Consider a highly skewed mixture distribution:
\[
X \sim \gamma \cdot \text{Beta}(40,1) + (1-\gamma)\cdot U[0,0.05], \quad \gamma = 0.01.
\]
With probability \(1-\gamma\), abilities are concentrated very close to zero, while a small fraction of contestants are drawn from a distribution sharply peaked near one. This creates a large gap between the extreme upper tail and the bulk of the population.

\textbf{Intuition.}
The mixture distribution places almost all contestants at very low ability while a tiny fraction are extreme outliers. For moderate \(n\), several high-types may coexist. With shortlist size \(m=2\), the cut-off is sometimes determined by another high-type, and its disclosure reveals that the remaining opponent is also strong. This sharp signal reduces uncertainty and allows the top contestant to lower effort.

In contrast, when \(m\) is slightly larger, the cut-off is almost surely near zero due to the mass of weak contestants. The disclosed score then carries little information about shortlisted rivals, increasing uncertainty and forcing the strongest contestant to exert higher effort. Hence larger \(m\) can outperform \(m=2\) in finite populations under extreme heterogeneity.

\begin{figure}[ht]
    \centering
    \includegraphics[width=\columnwidth]{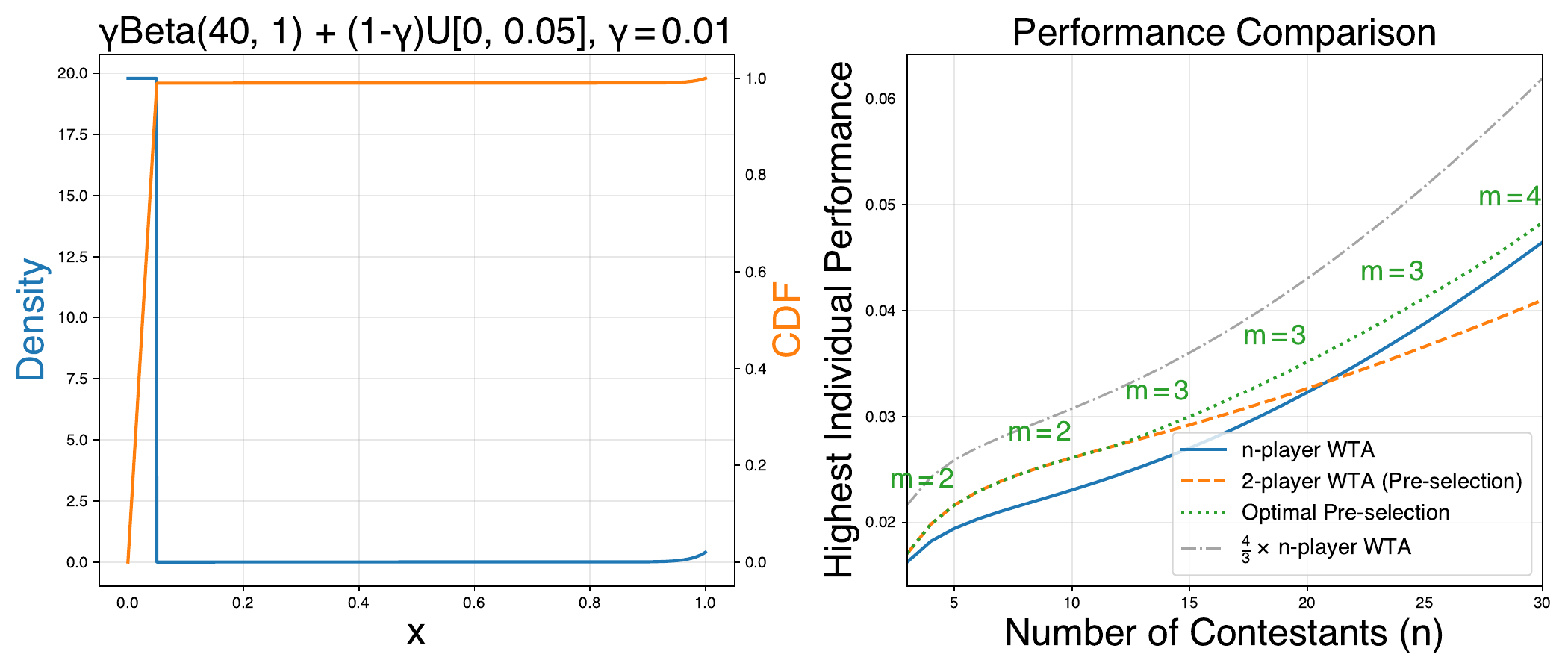}
    \caption{Highest individual performance under a skewed mixture ability distribution (\(X \sim 0.01,\text{Beta}(40,1) + 0.99,U[0,0.05]\)). In finite populations, shortlist sizes (\(m>2\)) can outperform the two-contestant WTA mechanism, illustrating distribution-dependent non-asymptotic optimality.}
    \label{fig:non-opt}
\end{figure}

\textbf{Numerical Illustration.}
Figure~\ref{fig:non-opt} plots the highest individual performance under different shortlist sizes. We observe that for moderate \(n\), the optimal mechanism shifts from \(m=2\) to \(m=3\) or \(m=4\), demonstrating that the universal optimality of the two-contestant WTA contest is distribution-dependent and may fail under special instance.

\newpage
\section*{NeurIPS Paper Checklist}

\begin{enumerate}

\item {\bf Claims}
    \item[] Question: Do the main claims made in the abstract and introduction accurately reflect the paper's contributions and scope?
    \item[] Answer: \answerYes{} 
    \item[] Justification: The abstract and introduction state the paper's theoretical scope: rank-order contests with transparent pre-selection, equilibrium characterization, optimal winner-take-all prize structures, the two-contestant asymptotic optimum for highest individual performance, and shortlist-size invariance for total performance. These claims are supported by the formal results in Sections 3--4 and complemented by numerical experiments in Section 5.

\item {\bf Limitations}
    \item[] Question: Does the paper discuss the limitations of the work performed by the authors?
    \item[] Answer: \answerYes{} 
    \item[] Justification: The paper discusses several limitations and future directions in Section 6, including more general cost functions, alternative objectives, noisy ability observation, and multi-contest extensions. The paper also notes that finite-population optimality of the two-contestant WTA contest is distribution-dependent, with a counterexample provided in Appendix C.


\item {\bf Theory assumptions and proofs}
    \item[] Question: For each theoretical result, does the paper provide the full set of assumptions and a complete (and correct) proof?
    \item[] Answer: \answerYes{} 
    \item[] Justification: The model assumptions are stated in Section 3, including the ability distribution, shortlisting rule, disclosed cut-off score, cost function, and equilibrium concept. The main text provides proof sketches for central results, while full proofs of the propositions, lemmas, and theorems are provided in the appendix.

    \item {\bf Experimental result reproducibility}
    \item[] Question: Does the paper fully disclose all the information needed to reproduce the main experimental results of the paper to the extent that it affects the main claims and/or conclusions of the paper (regardless of whether the code and data are provided or not)?
    \item[] Answer: \answerYes{}{} 
    \item[] Justification: The numerical experiments are based on fully specified synthetic ability distributions and contest configurations described in Section 5 and Appendix C. The plotted quantities follow directly from the equilibrium formulas and objective definitions provided in the paper, making the experiments reproducible from the stated mathematical specification.

\item {\bf Open access to data and code}
    \item[] Question: Does the paper provide open access to the data and code, with sufficient instructions to faithfully reproduce the main experimental results, as described in supplemental material?
    \item[] Answer: \answerNo{} 
    \item[] Justification: The experiments use synthetic distributions and analytically specified objectives, so the results can still be easily reproduced from the formulas and experimental descriptions in the paper.

\item {\bf Experimental setting/details}
    \item[] Question: Does the paper specify all the training and test details (e.g., data splits, hyperparameters, how they were chosen, type of optimizer) necessary to understand the results?
    \item[] Answer: \answerYes{} 
    \item[] Justification: The paper does not involve model training, data splits, optimizers, or hyperparameter selection. The experimental settings are numerical simulations based on specified ability distributions, shortlist sizes, and noisy observation models, as described in Section 5 and Appendix C.

\item {\bf Experiment statistical significance}
    \item[] Question: Does the paper report error bars suitably and correctly defined or other appropriate information about the statistical significance of the experiments?
    \item[] Answer: \answerNo{} 
    \item[] Justification: The current numerical figures do not report error bars, confidence intervals, or statistical significance tests. The experiments are intended as deterministic numerical illustrations of the theoretical formulas rather than empirical estimates from sampled datasets.

\item {\bf Experiments compute resources}
    \item[] Question: For each experiment, does the paper provide sufficient information on the computer resources (type of compute workers, memory, time of execution) needed to reproduce the experiments?
    \item[] Answer: \answerNo{} 
    \item[] Justification: The paper does not currently report the compute environment, memory usage, or runtime for the numerical experiments. The experiments are lightweight numerical evaluations of closed-form or integral expressions and do not require specialized hardware.
    
\item {\bf Code of ethics}
    \item[] Question: Does the research conducted in the paper conform, in every respect, with the NeurIPS Code of Ethics \url{https://neurips.cc/public/EthicsGuidelines}?
    \item[] Answer: \answerYes{} 
    \item[] Justification: The work is a theoretical and numerical study of contest design with transparent pre-selection. It does not involve human-subject experiments, private data, deployed systems, or high-risk data/model releases.

\item {\bf Broader impacts}
    \item[] Question: Does the paper discuss both potential positive societal impacts and negative societal impacts of the work performed?
    \item[] Answer: \answerNA{} 
    \item[] Justification:  This paper is primarily a theoretical study that characterizes the role of transparent pre-selection in rank-order contests. It does not propose a deployable system, release data or models, or make empirical claims about interventions on real contest platforms. Therefore, a dedicated broader-impact discussion is not directly applicable beyond the theoretical implications of the model.
    
\item {\bf Safeguards}
    \item[] Question: Does the paper describe safeguards that have been put in place for responsible release of data or models that have a high risk for misuse (e.g., pre-trained language models, image generators, or scraped datasets)?
    \item[] Answer: \answerNA{} 
    \item[] Justification: The paper does not release high-risk models, scraped datasets, or other assets with substantial misuse potential. The experiments are based on synthetic distributions and theoretical contest-design analysis.

\item {\bf Licenses for existing assets}
    \item[] Question: Are the creators or original owners of assets (e.g., code, data, models), used in the paper, properly credited and are the license and terms of use explicitly mentioned and properly respected?
    \item[] Answer: \answerNA{} 
    \item[] Justification: The paper does not rely on existing datasets, pretrained models, or external code assets as core research artifacts. The introductory figure uses aggregate statistics from public sources, which are cited in the figure caption.

\item {\bf New assets}
    \item[] Question: Are new assets introduced in the paper well documented and is the documentation provided alongside the assets?
    \item[] Answer: \answerNA{} 
    \item[] Justification: The paper does not introduce or release new datasets, models, benchmarks, or software assets. The numerical experiments are generated from the theoretical model and synthetic ability distributions described in the paper.

\item {\bf Crowdsourcing and research with human subjects}
    \item[] Question: For crowdsourcing experiments and research with human subjects, does the paper include the full text of instructions given to participants and screenshots, if applicable, as well as details about compensation (if any)? 
    \item[] Answer: \answerNA{} 
    \item[] Justification: The paper does not involve crowdsourcing experiments or human-subject research. References to contest and crowdsourcing platforms are used only as motivation for the theoretical model.

\item {\bf Institutional review board (IRB) approvals or equivalent for research with human subjects}
    \item[] Question: Does the paper describe potential risks incurred by study participants, whether such risks were disclosed to the subjects, and whether Institutional Review Board (IRB) approvals (or an equivalent approval/review based on the requirements of your country or institution) were obtained?
    \item[] Answer: \answerNA{} 
    \item[] Justification: The paper does not involve human-subject research, participant recruitment, interventions, or collection of individual-level human data. Therefore, IRB approval or equivalent review is not applicable.

\item {\bf Declaration of LLM usage}
    \item[] Question: Does the paper describe the usage of LLMs if it is an important, original, or non-standard component of the core methods in this research? Note that if the LLM is used only for writing, editing, or formatting purposes and does \emph{not} impact the core methodology, scientific rigor, or originality of the research, declaration is not required.
    \item[] Answer: \answerNA{} 
    \item[] Justification: The core method development, theoretical results, and numerical experiments do not use LLMs as an important, original, or non-standard research component. Any use of LLMs for writing assistance, editing, or formatting does not affect the paper's methodology, results, or scientific claims.

\end{enumerate}

\end{document}